\documentclass[conference]{IEEEtran}
\IEEEoverridecommandlockouts
\usepackage{cite}
\usepackage{amsmath,amssymb,amsfonts}
\usepackage{algorithmic}
\usepackage{graphicx}
\usepackage{textcomp}
\usepackage{subfiles}
\usepackage{xcolor}
\usepackage{xspace}
\usepackage{comment}
\usepackage{float}

\newcommand{\ie}{\emph{i.e.,}\xspace}
\newcommand{\aka}{\emph{a.k.a.}\xspace}
\newcommand{\eg}{\emph{e.g.,}\xspace}
\newcommand{\etal}{\emph{et al.}\xspace}
\def\BibTeX{{\rm B\kern-.05em{\sc i\kern-.025em b}\kern-.08em
    T\kern-.1667em\lower.7ex\hbox{E}\kern-.125emX}}

\usepackage{caption}

\usepackage{multirow}
\usepackage{array}
\usepackage{tikz}
\usepackage{listings}
\usepackage{courier}
\usepackage{enumitem}
\usepackage{times}
\usepackage{booktabs}
\usepackage{subfigure}
\usepackage[ruled,vlined,linesnumbered]{algorithm2e}
\usepackage{amsthm, bm}
\usepackage{makecell}

\newtheorem{theorem}{Theorem}

\newtheorem{definition}{Definition}
\newtheorem{example}{Example}
\newtheorem{lemma}{Lemma}

\newtheorem{observation}{Observation}

\newtheorem{remark}{Remark}

\newcommand{\WW}{\mathrm{\mathcal{W}}}

\newcommand{\DD}{\mathrm{\mathcal{D}}}
\def\cbl{}

\def\todo{\textcolor{red}}
\def\cthe{}
\newif\ifLONGVERSION
\LONGVERSIONtrue 

\begin{document}

\title{{DB-LSH}: Locality-Sensitive Hashing with Query-based Dynamic Bucketing 
}

\author{\IEEEauthorblockN{Yao Tian, Xi Zhao, Xiaofang Zhou}
\IEEEauthorblockA{
The Hong Kong University of Science and Technology\\
Hong Kong SAR, China \\
\{ytianbc@cse.ust.hk, xizhao@ust.hk, zxf@cse.ust.hk\}}
}
\maketitle

\begin{abstract}
Among many solutions to the high-dimensional approximate nearest neighbor (ANN) search problem, locality sensitive hashing (LSH) is known for its sub-linear query time and robust theoretical guarantee on query accuracy. Traditional LSH methods can generate a small number of candidates quickly from hash tables but suffer from large index sizes and hash boundary problems. Recent studies to address these issues often incur extra overhead to identify eligible candidates or remove false positives, making query time no longer sub-linear. To address this dilemma, in this paper we propose a novel LSH scheme called DB-LSH 
which supports efficient ANN search for large high-dimensional datasets.
It organizes the projected spaces with multi-dimensional indexes rather than using fixed-width hash buckets. Our approach can significantly reduce the space cost by avoiding the need to maintain many hash tables for different bucket sizes. During the query phase of DB-LSH, a small number of high-quality candidates can be generated efficiently by dynamically constructing query-based hypercubic buckets with the required widths through index-based window queries. For a dataset of $n$ $d$-dimensional points with approximation ratio $c$, our rigorous theoretical analysis shows that DB-LSH achieves a smaller query cost $\bm{O(n^{\rho^*} d\log n)}$, where $\bm{\rho^*}$ is bounded by $\bm{1/c^{\alpha}}$ 
versus a bound of $\bm{1/c}$ in the existing work. An extensive range of experiments on real-world data demonstrate the superiority of DB-LSH over state-of-the-art methods on both efficiency and accuracy.
\end{abstract}

\begin{IEEEkeywords}
Locality Sensitive Hashing, Approximate Nearest Neighbor Search, High-Dimensional Spaces
\end{IEEEkeywords}

\section{Introduction}
The nearest neighbor (NN) search finds the closest point in a point dataset to a given query point. As the points which are closer to each other can often be considered `similar' to each other in many applications when a proper distance measure is used, this search operation plays a vital role in a wide range of areas, such as pattern recognition\cite{DBLP:journals/scientometrics/AbdulhayogluT18}, information retrieval\cite{DBLP:conf/recsys/WinecoffBCWG19}, and data mining\cite{DBLP:conf/www/HeLZNHC17}.  However, it is well known that finding the exact NN in large-scale high-dimensional datasets can be very time-consuming. People often conduct approximate nearest neighbor (ANN) searches instead\cite{DBLP:conf/stoc/IndykM98,DBLP:conf/sigmod/TaoYSK09}. The $c$-approximate nearest neighbor ($c$-ANN) search and $(r, c)$-nearest neighbor ($(r, c)$-NN) search are two representative queries to trade result accuracy for query efficiency. Specifically,  $c$-ANN search aims to find a point whose distance to the query point $q$ is bounded by $cr^*$, where $r^*$ is the distance from $q$ to its exact NN and $c$ is a given approximation ratio  (see Definition \ref{def:cANN}, Section \ref{preliminaries}). $(r, c)$-NN search can be considered as a decision version of $c$-ANN, which aims to determine whether there exists a point whose distance to  $q$ is at most $cr$, where $r$ is a given search range (see Definition \ref{def:rcNN}, Section \ref{preliminaries}).

\begin{table*}[t]
	\centering
	\caption{Comparison of Typical LSH Methods}
	\label{tb:comp_ATF_alg}
	\resizebox{.9\textwidth}{!}
	{
	\begin{tabular}{|c|c|c|c|c|c|c|}
		\hline
		\multicolumn{2}{|c|}{\textbf{Algorithms}}                    & \textbf{Indexing} & \textbf{Query}& \textbf{Index Size} & \textbf{Query Cost}  & \textbf{Comment} \\
		\hline
		\multicolumn{1}{|c|}{\multirow{3}{*}{\textbf{KL}}} 
		&\textbf{DB-LSH}& \textbf{Dynamic}
		&  \textbf{Query-centric} &  \bm{$O(n^{1+\rho^*}d\log n)$} & \bm{$O(n^{\rho^*}d\log n)$} &  \bm{$\rho^*\le 1/c^{\alpha}$}\\
		\cline{2-7}
		\multicolumn{1}{|c|}{} &E2LSH \cite{e2lsh} & Static & Query-oblivious &  $O(Mn^{1+\rho}d\log n)$ & $O(n^{\rho}d\log n)$ &  $\rho\le 1/c$\\
		\cline{2-7}
		\multicolumn{1}{|c|}{} &LSB-Forest \cite{DBLP:conf/sigmod/TaoYSK09}& Static & Query-oblivious &  $O(n^{1+\rho}d\log n)$ & $O(n^{\rho}d\log n)$ &  $\rho\le 1/c,c\ge2$\\
		\hline
		\multicolumn{1}{|c|}{\multirow{3}{*}{\textbf{C2}}} &QALSH \cite{DBLP:journals/pvldb/HuangFZFN15} & Dynamic & Query-centric &  $O(nK)$ & $O(nK+d)$ &  $K=O(\log n)$\\
		\cline{2-7}
		\multicolumn{1}{|c|}{} &VHP \cite{DBLP:journals/pvldb/LuWWK20}& Dynamic & Query-centric &  $O(nK)$ & $O(n(K+d))$ &  $K=O(1)$\\
		\cline{2-7}
		\multicolumn{1}{|c|}{} &R2LSH \cite{DBLP:conf/icde/LuK20}& Dynamic & Query-centric &  $O(nK)$ & $O(n(K+d))$ &  $K=O(1)$\\
		\hline
		\multicolumn{1}{|c|}{\multirow{2}{*}{\textbf{MQ}}} & SRS \cite{DBLP:journals/pvldb/SunWQZL14}& Dynamic & Query-centric &  $O(n)$ & $O(\beta n(\log n+d))$ &  $\beta\ll 1$\\
		\cline{2-7}
		\multicolumn{1}{|c|}{} &PM-LSH \cite{DBLP:journals/pvldb/ZhengZWHLJ20}& Dynamic & Query-centric &  $O(n)$ & $O(\beta nd)$ &  $\beta\ll 1$\\
		\hline
	\end{tabular}}
\end{table*}

Locality-Sensitive Hashing (LSH) \cite{DBLP:conf/sigmod/GanFFN12,DBLP:conf/sigmod/GaoJLO14,DBLP:conf/sigmod/TaoYSK09,DBLP:conf/sigmod/ZhengGTW16,DBLP:conf/vldb/GionisIM99,e2lsh} is one of the most popular tools for computing $c$-ANN in high-dimensional spaces. 
LSH maps data points into buckets using a set of hash functions such that nearby points in the original space have a higher probability to be hashed into the same bucket than those which are far away. When a query arrives, the probability to find its $c$-ANN is guaranteed to be sufficiently high by only checking the points in the bucket where the query point falls in. 
In order to achieve this goal, 
the original LSH-based methods (E2LSH) \cite{e2lsh} 
design a set of $K$ independent hash functions with which all data points in the original $d$-dimensional space are mapped into a $K$-dimensional space, $K\ll d$. These $K$-dimensional points are assigned into a range of buckets which are $K$-dimensional hypercubes. This process is repeated $L$ times to generate $L$  $K$-dimensional hash buckets (we term this type of approach $(K, L)$-index). Intuitively, as $K$ increases, the probability of two different points being hashed into the same bucket decreases. On the contrary, the collision probability, which is the probability of two different points being mapped into the same bucket, increases as $L$ increases because two points are considered as \cthe{a} ‘collision’ as long as they are mapped into the same bucket at least once. 
As shown in \cite{DBLP:conf/vldb/GionisIM99, DBLP:conf/compgeom/DatarIIM04}, by choosing $K=\log_{1/p_2}n$ and $L=n^\rho$, where $\rho = \frac{\ln 1/p_1}{\ln 1/p_2}$, $p_1, p_2$ are constants depending on $r$ and $c$ (for the meaning of $p_1$ and $p_2$, see Definition \ref{def:LSH}, Section \ref{preliminaries}), E2LSH can solve \cbl{the} $(r, c)$-NN problem in sub-linear time $O(n^\rho d\log n)$ with constant success probability of $1/2-1/e$. Accordingly, E2LSH finds $c$-ANN in sub-linear time by answering a series of $(r,c)$-NN queries with $r = 1, c, c^2, \ldots$. However, to achieve a good accuracy, E2LSH needs to prepare a $(K, L)$-index for each $(r, c)$-NN and $L$ is typically large, which causes prohibitively large storage costs for the indexes. LSB \cite{DBLP:conf/sigmod/TaoYSK09} alleviates this issue by building a $(K, L)$-index for $(1, c)$-NN and 
repeatedly merging small hash buckets into a large one, which effectively enlarges $r$. However, LSB only works for $(r,c)$-NN queries at some discrete integer $r$, which imposes the limitation that LSB cannot answer \cthe{the} $c$-ANN query with $c < 4$. C2LSH \cite{DBLP:conf/sigmod/GanFFN12} proposes a new LSH scheme called \textit{collision counting} (C2). By relaxing the collision condition from the exactly $K$ collisions to any $l$ collisions where $l< K$ is a given value, C2LSH only needs to maintain $K$ one-dimensional hash tables (instead of $L$ $K$-dimensional hash tables). However, the query cost of C2 is no longer sub-linear \cite{DBLP:conf/sigmod/GanFFN12} because it is expensive to count the number of collisions between a large number of data points and the query point dimension by dimension.

In addition to the dilemma between space and time, the above methods also suffer from the candidate quality issue (\aka the hash boundary issue). That is, no matter how large the hash buckets are, some points close to a query point may still be partitioned into different buckets. Several dynamic bucketing techniques are proposed to address this issue. The main idea of dynamic bucketing is to leave the bucketing process to the query phase in the hope of generating buckets such that the nearby points are more likely to be in the same bucket as the query point. The C2 approach is extended to dynamic scenarios by using B$^+$-trees to locate points falling in a query-centric bucket in each dimension \cite{DBLP:journals/pvldb/HuangFZFN15,DBLP:journals/pvldb/LuWWK20,DBLP:conf/icde/LuK20}, at the cost of increased query time because of a large number of one-dimensional searches.
\cite{DBLP:journals/pvldb/SunWQZL14, DBLP:journals/pvldb/ZhengZWHLJ20} explore a new dynamic metric query (MQ) based LSH scheme to map data points in a high-dimensional space into a low-dimensional projected space via $K$ independent LSH functions, and determine $c$-ANN by exact nearest neighbor searches in the projected space. However, even in a low-dimensional space, finding the exact NN is still inherently computationally expensive. More importantly, at least $\beta n$ exact distance computations are needed to perform in case of missing correct $c$-ANN, which incurs a linear time complexity. Here $\beta$ is an estimated ratio for 
the number of $K$ dimensional NN searches such that the $d$ dimensional ANN results can be found safely \cite{DBLP:journals/pvldb/SunWQZL14,DBLP:journals/pvldb/ZhengZWHLJ20}. 

Table 1 summarizes \cbl{the} query and space costs of typical LSH methods. As shown in the table, among the existing solutions to the $c$-ANN search problem, $(K, L)$-index based methods are the only ones that can achieve sub-linear query cost, \ie $O(n^\rho d\log n)$, where $\rho$ is proven to be bounded by $1/c$. $M$ in E2LSH is the number of $(K, L)$-indexes prepared ahead \cite{DBLP:conf/sigmod/TaoYSK09}. 
Note that the value of $\rho$ is bounded by $1/c$ only when the bucket size is very large \cite{DBLP:conf/compgeom/DatarIIM04}. This implies a very large value of $K$ is necessary to effectively differentiate points based on their distances. It remains a significant challenge to find \cbl{a} smaller and truly bounded $\rho$ without using a very large bucket size.

Motivated by the aforementioned limitations, in this paper we propose a novel $(K, L)$-index approach with a query-centric dynamic bucketing strategy called DB-LSH to solve the high-dimensional $c$-ANN search problem. DB-LSH decouples the hashing and bucketing processes of \cbl{the} $(K, L)$-index, making it possible to answer $(r,c)$-NN queries for any $r$ and $c$-ANN for any $c > 1$ with only one suit of indexes (\ie without the need to perform LSH $L$ times for each possible $r$). In this way the space cost is reduced significantly, and a reduction of $L$ value becomes possible.
DB-LSH builds dynamic query-centric buckets and conducts multi-dimensional window queries to eliminate the hash  boundary issues for selecting the candidates. Different from other query-centric methods, the region of our buckets are still multi-dimensional cubes like in static $(K, L)$-index methods, which enables DB-LSH to not only generate high-quality candidates but also to achieve sub-linear query cost, as shown in Table \ref{tb:comp_ATF_alg}.
Furthermore,  DB-LSH achieves a much smaller bound at a proper and finite bucket size, denoted as $\rho^*$, which is  bounded by $1/c^{\alpha}$  (\eg $\alpha=4.746$ when choosing $4c^2$ as the width of the initial hypercubic bucket). With theoretical analysis and an extensive range of experiments, we
show that DB-LSH outperforms the existing LSH methods significantly for both efficiency and accuracy.

The main contributions of this paper include:
\begin{itemize}
\item We propose a novel LSH framework, called DB-LSH, to solve \cthe{the} high-dimensional $c$-ANN search problem. It is the first work that combines the static $(K, L)$-index approach with a dynamic search strategy for bucketing. By taking advantages from both sides, DB-LSH can reduce the index size and improve query efficiency simultaneously.

\item A rigorous theoretical analysis shows that DB-LSH can achieve the lowest query time complexity so far for any approximation ratio $c>1$.  DB-LSH answers a $c^2$-ANN query with a constant success probability in $O(n^{\rho^*} d\log n)$ time, where $\rho^*$ is bounded by $1/c^{\alpha}$, \eg $\alpha=4.746$ when initial bucket width is $4c^2$, which is smaller than $\rho$ in other $(K,L)$-index methods. 

\item Extensive experiments on 10 real datasets with different sizes and dimensionality have been conducted to show that DB-LSH can achieve better efficiency and accuracy than the existing LSH methods.
\end{itemize}

The rest of the paper is organized as follows. The related work is reviewed in Section \ref{sec:relatedwork}. Section \ref{sec:preliminary} introduces the basic concepts and the research problem formally. The construction and query algorithms of DB-LSH are presented in Section \ref{sec:method}, with a theoretical analysis in Section \ref{sec:theory} and an experimental study in \ref{sec:experiment}. We conclude this paper in Section \ref{sec:conclusion}.

\section{Related Work} \label{sec:relatedwork}
\label{relatedwork}
LSH is originally proposed in \cite{DBLP:conf/stoc/IndykM98,DBLP:conf/vldb/GionisIM99}. Due to its simple structure, sub-linear query cost, and a rigorous quality guarantee, it has been 
a prominent approach for processing approximate nearest neighbor queries in \cthe{the} high dimensional spaces \cite{DBLP:conf/vldb/GionisIM99, DBLP:conf/compgeom/DatarIIM04, DBLP:conf/www/BawaCG05, DBLP:conf/vldb/LvJWCL07}. We give a brief overview of the existing LSH methods in this section.

\subsection{Mainstream LSH Methods}
\noindent
\textbf{$\bm{(K, L)}$-index based methods.} 
Although the basic LSH \cite{DBLP:conf/vldb/GionisIM99} is used in \cthe{the} Hamming space, $(K, L)$-index methods extend from it to provide a universal and well-adopted LSH framework for answering \cthe{the} $c$-ANN problem in other metric spaces. 
E2LSH\cite{e2lsh} is a popular $(K, L)$-index method in \cthe{the} Euclidean space and adopts the $p$-stable distribution-based function proposed in \cite{DBLP:conf/compgeom/DatarIIM04} as the LSH function. Its applications are limited by the hash boundary problem and undesirably large index sizes. These two shortcomings are shared by other $(K, L)$-index methods due to the fact that static buckets are used in these methods. To reduce index sizes, Tao et al. \cite{DBLP:conf/sigmod/TaoYSK09} consider answering $(r,c)$-NN queries at different radii via an elegant LSB-Tree framework, although it only works for $c$-ANN query with $c\ge 4$. SK-LSH \cite{DBLP:journals/pvldb/LiuCHLS14} is another approach based on the idea of static $(K, L)$-index, but proposes a novel search framework to find more candidates.

To address the limitations of static $(K, L)$-index methods, dynamic query strategies are developed to find high-quality candidates using smaller indexes. These methods can be classified into two categories as follows.

\noindent
\textbf{Collision counting based methods (C2).} 
The core idea of C2 is to generate candidates based on the collision numbers. It is proposed in C2LSH \cite{DBLP:conf/sigmod/GanFFN12}, which uses the techniques of collision counting and virtual rehashing to reduce space consumption. 
QALSH \cite{DBLP:journals/pvldb/HuangFZFN15} improves C2LSH by adopting query-aware buckets rather than static ones, which alleviates the hash boundary issue.
R2LSH \cite{DBLP:conf/icde/LuK20} improves the performance of QALSH by mapping data into multiple two-dimensional projected spaces rather than one-dimensional projected spaces as in QALSH.
VHP \cite{DBLP:journals/pvldb/LuWWK20} considers the buckets in QALSH as hyper-planes and introduces the concept of virtual hyper-sphere to achieve smaller space complexity than QALSH.
C2 can find high-quality candidates with a larger probability but its cost of finding the candidates is expensive due to the unbounded search regions, which makes all points likely to be counted once in the worst case.

\noindent
\textbf{Dynamic metric query based methods (MQ).} 
SRS \cite{DBLP:journals/pvldb/LiuCHLS14} and PM-LSH \cite{DBLP:journals/pvldb/ZhengZWHLJ20} are representative dynamic MQ approaches that map data into \cbl{a} low-dimensional projected space and determine candidates based on their Euclidean distances  via queries in the projected space. It is proven that this strategy can accurately estimate the distance between two points in high-dimensional spaces \cite{DBLP:journals/pvldb/ZhengZWHLJ20}. However, answering metric queries in the projected space is still computationally expensive and as many as $\beta n$ candidates have to be checked to ensure a success probability of $1/2-1/e$, where $\beta$ is a constant mentioned earlier. Therefore, MQ can incur a high query cost of $\beta nd$.

\subsection{Additional LSH Methods}
There are other LSH methods that come from two categories: the methods that design different hash functions and the methods that adopt alternative query strategies.
The former includes studies that aim to propose novel LSH functions in Euclidean space with smaller $\rho$ \cite{DBLP:conf/focs/AndoniI06,DBLP:conf/nips/AndoniILRS15,DBLP:conf/stoc/AndoniR15}. However, these functions are highly theoretical and difficult to use. 
The latter focuses on finding better query strategies to further reduce the query time or index size \cite{DBLP:conf/www/BawaCG05, DBLP:conf/vldb/LvJWCL07, DBLP:conf/sigmod/LeiHKT20,DBLP:conf/soda/Panigrahy06, DBLP:journals/pvldb/SatuluriP12,DBLP:conf/sigmod/ZhengGTW16, DBLP:conf/icde/LiuWZWQ19, liu2020ei,DBLP:conf/sigmod/LeiHKT20}.
LSH forest \cite{DBLP:conf/www/BawaCG05} offers each point a variable-length hash value instead of a fixed $K$ hash value as in $(K, L)$-index methods. It can improve the quality guarantee of LSH for skewed data distributions while retaining the same space consumption and query cost.
Multi-Probe LSH \cite{DBLP:conf/vldb/LvJWCL07} examines multiple hash buckets in the order of a probing sequence derived from a hash table. 
It reduces the space requirement of E2LSH at the cost of the quality guarantee.
Entropy-based LSH \cite{DBLP:conf/soda/Panigrahy06} and BayesLSH \cite{DBLP:journals/pvldb/SatuluriP12} adopt similar multi-probing strategies as in Multi-Probe LSH, but have a more rigorous theoretical analysis. Their theoretical analysis relies on a strong assumption on data distribution which can be hard to satisfy, leading to poor performance for some datasets. 
LazyLSH \cite{DBLP:conf/sigmod/ZhengGTW16} supports $c$-ANN queries in multiple $p$-norm spaces with only one suit of indexes, thus effectively reducing the space consumption. 
I-LSH \cite{DBLP:conf/icde/LiuWZWQ19} and EI-LSH \cite{liu2020ei} design a set of adaptive early termination conditions so that the query process can stop early if a good enough result is found.
Developed upon SK-LSH \cite{DBLP:journals/pvldb/LiuCHLS14} and Suffix Array \cite{DBLP:conf/soda/ManberM90}, Lei \etal \cite{DBLP:conf/sigmod/LeiHKT20} propose a dynamic concatenating search framework, LCCS-LSH, that also achieves sub-linear query time and sub-quadratic space.

Recently, researchers have adopted the LSH framework to solve other kinds of queries, such as maximum inner product search\cite{DBLP:conf/nips/Shrivastava014,DBLP:conf/icml/NeyshaburS15,DBLP:conf/kdd/HuangMFFT18,DBLP:conf/nips/YanLDCC18} and point-to-hyperplane NN search \cite{DBLP:conf/sigmod/HuangLT21} in high dimensional spaces. These examples demonstrate the superior performance and great scalability of LSH.

\section{Preliminaries} \label{sec:preliminary}
\label{preliminaries}
\cbl{In this section, we present the definition of the ANN search problem, the concepts of LSH, and an important observation.} Frequently used notations are summarized in Table \ref{tab:Commonly}.

\begin{table}[t]
    \centering
    \caption{List of Key Notations.}
    \label{tab:Commonly}
    \setlength{\tabcolsep}{7mm}{
    \begin{tabular}{|c|c|}
        \hline
        \textbf{Notation} & \textbf{Description}\\ 
        \hline 
        $\mathbb{R}^d$ & $d$-dimensional Euclidean space\\
        \hline
        $\mathcal{D}$ & The dataset\\
        \hline
        $n$ &  The cardinality of dataset\\
        \hline 
        $o$ & A data point\\ 
        \hline
        $q$ & A query point\\
        \hline 
        $\|o_1,o_2\|$ & The distance between $o_1$ and $o_2$ \\
        \hline 
        $f(x)$ & The pdf of standard normal distribution\\
        \hline
        $h(x)$ & Hash function\\
        \hline 
    \end{tabular}}
\end{table}

\subsection{Problem Definitions}
Let $\mathbb{R}^d$ be a $d$-dimensional Euclidean space, and $\|\cdot,\cdot\|$ denote the distance between points.

\begin{definition}[$c$-ANN Search]
\label{def:cANN}
Given a dataset $\mathcal{D} \subseteq  \mathbb{R}^d$, a query point $q \in \mathbb{R}^d$ and an approximation ratio $c > 1$, $c$-ANN search returns a point $o \in \mathcal{D}$ satisfying $\|q, o\| \leq c\cdot\|q, o^*\|$, where $o^*$ is the exact nearest neighbor of $q$.
\end{definition}

\begin{remark}
$(c,k)$-ANN search is a natural generalization of $c$-ANN search. It returns $k$ points, say $o_1, \ldots, o_k$ that are sorted in ascending order w.r.t. their distances to $q$, such that for $\forall o_i, i = 1, \ldots, k$, we have $ \|q, o_i\| \leq c\cdot\|q, o_i^*\|$, where $o_i^*$ is the $i$-th nearest neighbor of $q$.
\end{remark}



$(r, c)$-nearest neighbor search is often used as a subroutine when finding $c$-ANN. Following \cite{DBLP:conf/sigmod/TaoYSK09}, it is defined formally as follows:

\begin{definition} [$(r, c)$-NN Search]
\label{def:rcNN}
Given a dataset $\mathcal{D} \subseteq \mathbb{R}^d$, a query point $q \in \mathbb{R}^d$, an approximation ratio $c > 1$ and a distance $r$,  $(r, c)$-NN search returns:
\begin{itemize}
    \item [(1)] a point  $o \in \mathcal{D}$ satisfying $\|q, o\| \leq c \cdot r$, if there exists a point $o^\prime \in \mathcal{D}$ such that $\|q, o^\prime\| \leq r$;
    \item [(2)] nothing, if there is no point $o \in \mathcal{D}$ such that $\|q, o\| \leq c \cdot r$.
    \item [(3)] otherwise, the result is undefined.
\end{itemize}
\end{definition}

The result of case 3 remains undefined since case 1 and case 2 suffice to ensure the correctness of a $c$-ANN query.
By setting $r = \|q, o^*\|$, where $o^*$ is the nearest neighbor of $q$, a $c$-ANN can be found directly by answering \cbl{an} $(r, c)$-NN query. As $\|q, o^*\|$ is not known in advance,
a $c$-ANN query is processed by conducting a series of $(r, c)$-NN queries with increasing radius, \ie it begins by searching a region around $q$ using a small $r$ value. Without loss of generality, we assume $r = 1$. Then, it keeps enlarging the search radius in multiples of $c$, \ie $r = c, c^2, c^3, \ldots$ until a point is returned. In this way, as shown in \cite{DBLP:conf/stoc/IndykM98, DBLP:conf/vldb/GionisIM99,e2lsh}, a $c$-ANN query can be answered with an approximation ratio of $c^2$.

\begin{example}
\label{eg:rcNN-cANN}
Figure \ref{fig: rcNN-cANN} shows an example where $\mathcal{D}$ has 12 data points. Suppose approximation ratio $c = 1.5$. Consider \cbl{the} first $(r, c)$-NN search with $r = 1$ (\cbl{the} yellow circle). Since there is no point $o \in \mathcal{D}$ such that $\|q, o\| \leq cr=1.5$ (\cbl{the} red circle), it returns nothing. Then, consider $(r, c)$-NN with $r = c = 1.5$. Since there exists no point $o$ such that $\|q, o\| \leq r$, but $\|q, o_4\| \leq cr$ (\cbl{the} blue circle), the returned result is undefined, \ie it is correct to return either nothing or any found point, such as $o_4$. Finally, consider $(r, c)$-NN with $r = c^2 = 2.25$. Since $\|q, o_4\| \leq 2.25$, the query must return a point, which can be 
any point from $o_4, o_6, o_9, o_{11}$ as all of them satisfy $\|q, o\| \leq cr$ (\cbl{the} green circle). The above procedures also elaborate the process of answering a $c^2$-ANN query. Any point from $o_4, o_6, o_9, o_{11}$ can be considered as a result. Apparently, they are correct $c^2$-ANN results of $q$.

\begin{figure}[htbp]
  \centering
  \includegraphics[width=0.46\linewidth]{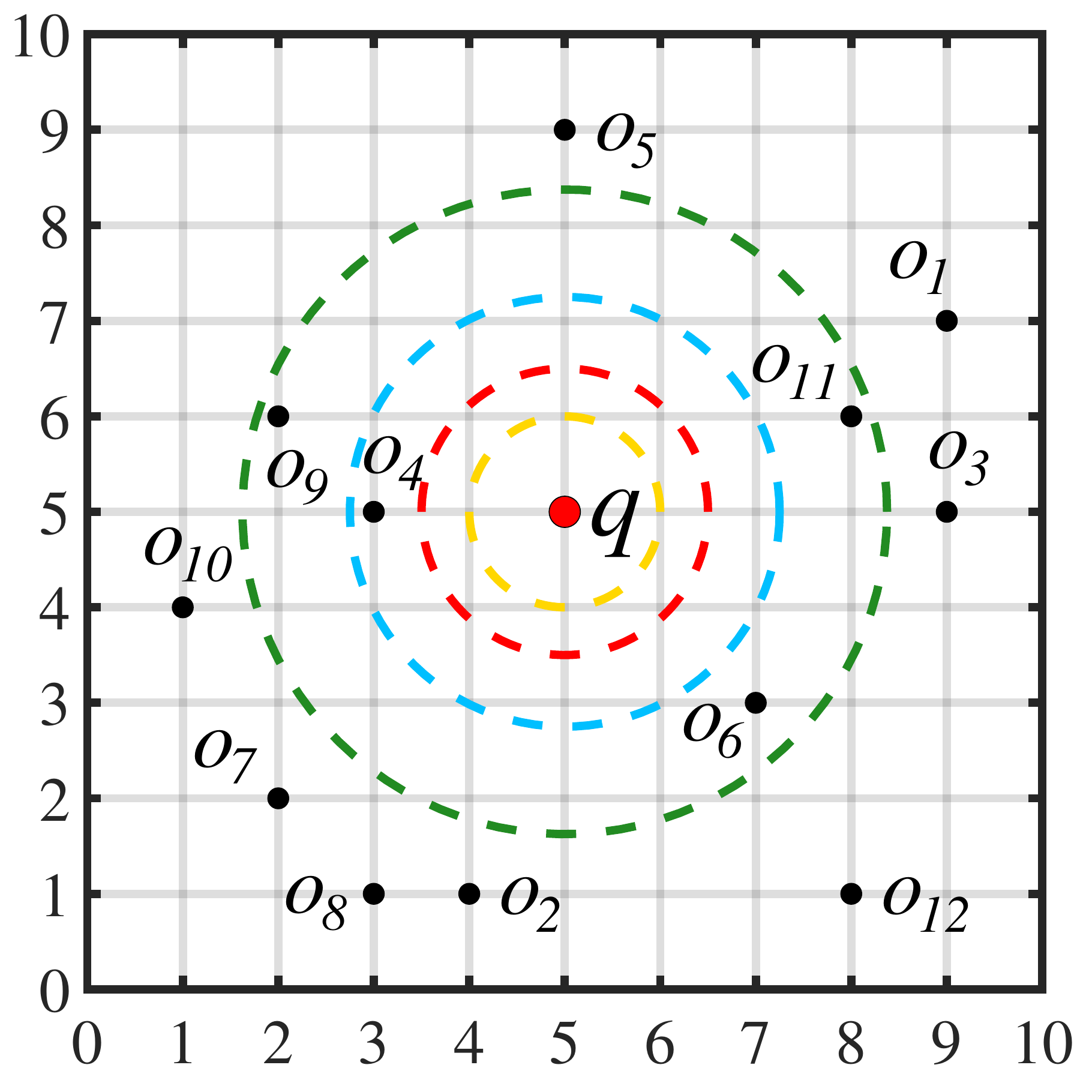}
  \caption{An illustration of $(r, c)$-NN and $c$-ANN }
  \label{fig: rcNN-cANN}
\end{figure}

\end{example}

\subsection{Locality-Sensitive Hashing}
Locality-sensitive hashing  is the foundation of our method. 
For a hash function $h$, two points $o_1$ and $o_2$ are said to collide over $h$ if $h(o_1) = h(o_2)$, \ie they are mapped into the same bucket using $h$. 
The formal definition of LSH is given below\cite{DBLP:conf/vldb/GionisIM99}:

\begin{definition}[LSH]
\label{def:LSH}
Given a distance $r \geq 0$ and an approximation ratio $c > 1$, a family of hash functions $\mathcal{H} = \{h:\mathbb{R}^d \rightarrow \mathbb{R}\}$ is called $(r, cr, p_1, p_2)$-locality-sensitive, if for $\forall o_1, o_2 \in \mathbb{R}^d$, it satisfies both conditions below:
\begin{itemize}
\item[(1)] If $\|o_1, o_2\| \leq r, \Pr[h(o_1) = h(o_2)] \geq p_1$;
\item[(2)] If $\|o_1, o_2\| > cr, \Pr[h(o_1) = h(o_2)] \leq p_2$, 
\end{itemize}
where $h\in \mathcal{H}$ is chosen at random, $p_1, p_2$ are collision probabilities and $p_1 > p_2$.
\end{definition}

A typical LSH family for Euclidean space
in  static LSH methods (\eg E2LSH) is defined as follows \cite{DBLP:conf/compgeom/DatarIIM04}:
\begin{equation}
\label{eq:lsh}
    h(o) =\left \lfloor \frac{\vec{a} \cdot \vec{o} + b}{w}  \right \rfloor,
\end{equation}
where $\vec{o}$ is the vector representation of a point $o\in \mathbb{R}^d$, $\vec{a}$ is a $d$-dimensional vector where each entry is chosen independently from a $2$-stable distribution, \ie \cbl{the} standard normal distribution, $b$ is a real number chosen uniformly from $[0, w)$, and $w$ is a pre-defined integer. 
Denote the distance between any two points as $\tau$, then the collision probability under such hash function can be computed as:
\begin{equation}
    p(\tau; w) = \Pr[h(o_1) = h(o_2)] = 2\int_0^w \frac{1}{\tau}\cdot f(\frac{t}{\tau})\cdot(1-\frac{t}{w})\,dt,
\end{equation}
where $f(x) = \frac{1}{\sqrt{2\pi}}e^{-\frac{x^2}{2}}$ is the probabilistic density function (pdf) of the standard normal distribution. For a given $w$, it is easy to see that $p(\tau; w)$ decreases monotonically with $\tau$. Therefore, the hash family defined by Equation \ref{eq:lsh} is $(r, cr, p_1, p_2)$-locality-sensitive, where $p_1 = p(r; w)$ and $p_2 = p(cr; w)$.
\subsection{Locality-Sensitive Hashing with Dynamic Bucketing}
A typical dynamic LSH family for \cthe{the} Euclidean space is defined as follows \cite{DBLP:journals/pvldb/HuangFZFN15}:

\begin{equation}
\label{eq:lsh2}
    h(o) = \vec{a} \cdot \vec{o},
\end{equation}
where $\vec{a}$ is the same as in Equation \ref{eq:lsh}. For a hash function $h$, two points $o_1$ and $o_2$ are said to collide over $h$ if $|h(o_1) - h(o_2)|\leq \frac{w}{2}$. In this sense, the collision probability can be computed as:
\begin{equation}
\label{eq:cp2}
    p(\tau; w) = \Pr[|h(o_1) - h(o_2)|\leq \frac{w}{2}] = \int_{-\frac{w}{2\tau}}^{\frac{w}{2\tau}}f(t)\,dt,
\end{equation}
It is easy to see that the hash family defined by Equation \ref{eq:lsh2} is $(r, cr, p_1, p_2)$-locality-sensitive, where $p_1 = p(r; w)$ and $p_2 = p(cr; w)$. 
In what follows, $\mathcal{H} = \{h: \mathbb{R}^d \to \mathbb{R}\} $ refers to the LSH family identified by Equation \ref{eq:lsh2} and $p(\tau; w)$ refers to \cbl{the} corresponding collision probability in Equation \ref{eq:cp2} unless otherwise stated. 

Next, we introduce a simple but important observation that inspires us to design a dynamic $(K, L)$-index.


\begin{observation}
The hash family is $(r, cr, p(1, w_0), p(c, w_0))$-locality-sensitive for any search radius $r$ and $w = rw_0$, where $w_0$ is a positive constant.
\end{observation}
\label{observation1}
\begin{proof}
It is easy to see that for any search radius $r$ and $w = rw_0$, the following equation holds:
\begin{equation}
    p(r; w_0r)=\int_{-\frac{w_0r}{2r}}^{\frac{w_0r}{2r}}f(t)\,dt=\int_{-\frac{w_0}{2}}^{\frac{w_0}{2}}f(t)\,dt=p(1;w_0).
\end{equation}
That is, $\mathcal{H}$ is $(r, cr, p(1, w_0), p(c, w_0))$-locality-sensitive.
\end{proof}
By the above observation, we do not need to physically maintain multiple $(K, L)$-indexes from $(r, cr, p(r, w), $ $p(cr, w))$-locality-sensitive hash family  in advance to support the corresponding $(r, c)$-NN queries with different $r$. Instead, we can dynamically partition buckets with the width required by different queries via only one $(K, L)$-index, where $K = \log_{1/p(c; w_0)}(\frac{n}{t})$, $L = (\frac{n}{t})^{\rho^*}$, $\rho^* = \frac{\ln{1/p(1; w_0)}}{\ln{1/p(c; w_0)}} ~\text{and}~ \cbl{t}$ \cbl{is a constant to balance the query efficiency and space consumption (see Remark \ref{remark:t}, Section \ref{theory})}. As explained in Section \ref{theory}, the choice of $K$ and $L$ guarantees correctness of DB-LSH for $(r,c)$-NN search and $c$-ANN search. This is a key observation that leads to our novel approach to be presented next.

\section{Our Method} \label{sec:method}
\label{method}

DB-LSH consists of an indexing phase for mapping and a query phase for dynamic bucketing. We first give an overview of this novel approach, followed by detailed descriptions of the two separate phases. 


\subsection{Overview of DB-LSH}
Considering the limitations of C2 and MQ discussed earlier, we propose to keep the basic idea of \cbl{the} static $(K, L)$-index, which provides an opportunity to answer $c$-ANN queries with \cthe{the} sub-linear query cost. To remove the inherent obstacles 
in static $(K, L)$-index methods, DB-LSH develops a dynamic bucketing strategy that constructs query-centric hypercubic buckets with the required width in the query phase.
In the indexing phase, DB-LSH projects each data point into $L$ $K$-dimensional spaces by $L \times K$ independent LSH functions. Unlike static $(K, L)$-index methods that quantify the projected points with a fixed size, we index points in each $K$-dimensional space with a multi-dimensional index. In the query phase, an $(r, c)$-NN query with sufficiently small $r$, say $r = 1$, is issued at the beginning. To answer this query,  $L$ query-centric hypercubic buckets with width $w_0$ are constructed and the points in them are found by window queries.
If the retrieved point is within $cr$ of $q$, DB-LSH returns it as a correct $c$-ANN result. Otherwise, the next $(r, c)$-NN query with $r = c$ is issued, and the width of the dynamic hypercubic bucket $w$ is updated from $w_0$ to $cw_0$ accordingly. By gradually extending the search radii $r = c^2, c^3 \ldots$ and bucket  width $w = w_0r$, DB-LSH achieves finding  $c$-ANN  with a constant success probability on top of just one $(K, L)$-index after accessing a maximum of $2tL + 1$ points. 



\begin{figure}[htbp]
  \centering
  \includegraphics[width=0.46\linewidth]{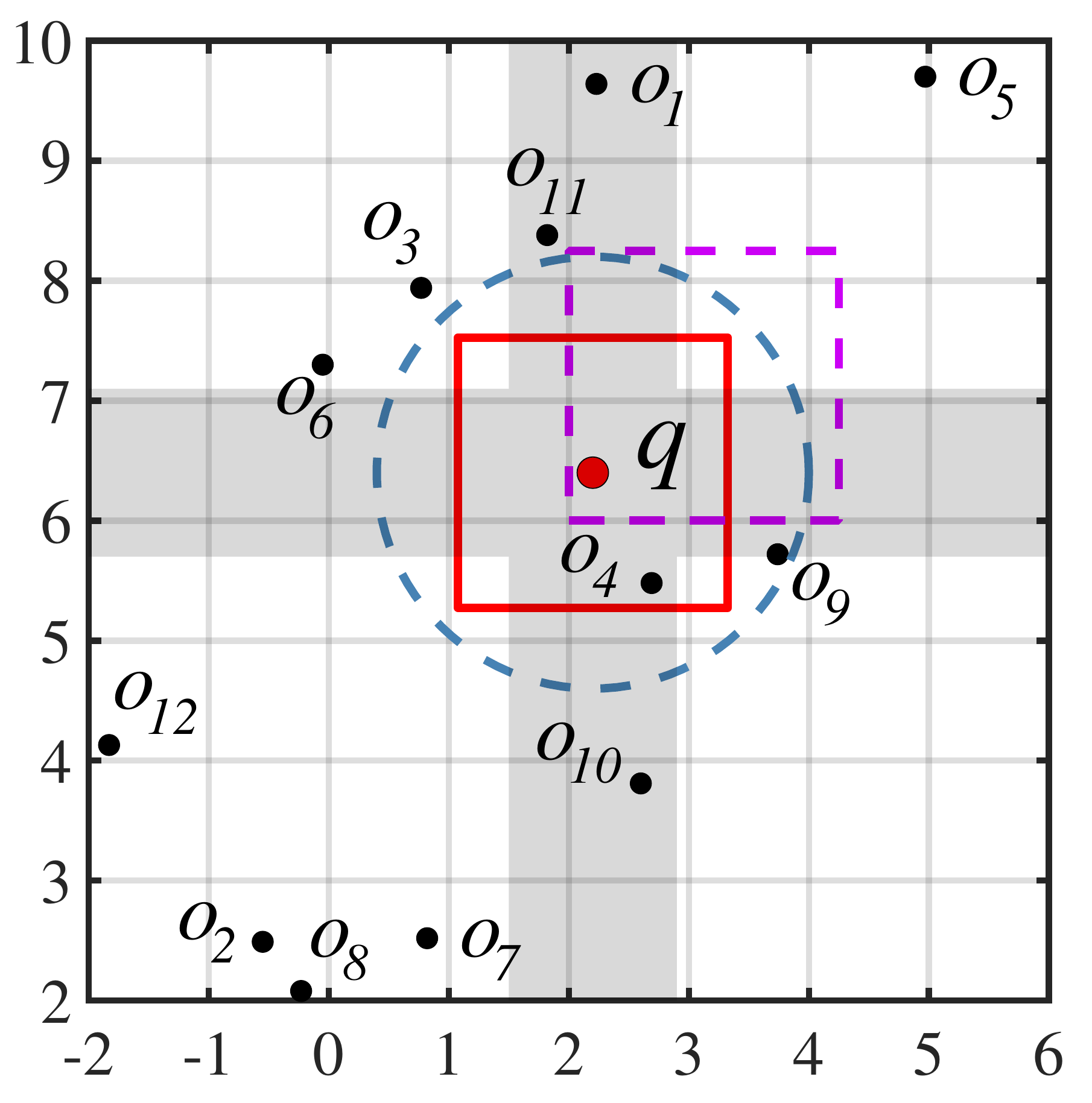}
  \caption{Search regions of DB-LSH and other LSH methods}
  \label{fig: searchregion}
\end{figure}

Figure \ref{fig: searchregion} gives an intuitive explanation of the advantages of DB-LSH on the search region.
\cbl{The} dotted purple square is the search region in E2LSH. We can notice that points close to the query might be hashed to a different bucket (\eg $o_4$), especially when $q$ is near to the bucket boundary, which jeopardizes the accuracy. \cbl{The} gray cross-like region is the search region of C2. Such an unbounded region is much bigger than that of DB-LSH (\cbl{the} red square), which leads to the number of points accessed arbitrarily large in the worst case and thus incurs a large query cost. \cbl{The} dotted blue circle is the search region of MQ. Although it is a bounded region, finding the points in it becomes more complex than in other regions. 
DB-LSH still uses hypercubic buckets (search region) as used in static $(K,L)$-index methods, but achieves much better accuracy. \cbl{The} query-centric bucketing strategy eliminates the hash boundary issue. The overhead of dynamic bucketing is affordable because of efficient window queries via multi-dimensional indexes.

To summarize, DB-LSH is hopeful of reaching a given accuracy with the least query cost among all these methods.
In what follows, we give everything that a practitioner needs to know to apply DB-LSH.

\subsection{Indexing Phase}
\label{index}
\cbl{The} indexing phase consists of two steps: constructing projected spaces and indexing points by multi-dimensional indexes.

\noindent \textbf{Constructing projected  spaces.} 
Given a $(1, c, p_1, p_2)$-locality-sensitive hash family $\mathcal{H}$, let $\mathcal{G}$ be the set of all subsets with $K$ hash functions chosen independently from $\mathcal{H}$, \ie each element $G \in \mathcal{G}$ is a $K$-dimensional compound hash of the form: 
\begin{equation}
    G(o) = (h_1(o), h_2(o), \ldots, h_K(o)),
\end{equation}
where $h_j \stackrel{\text{i.i.d}}{\sim}\mathcal{H}, j = 1, \ldots, K$.
Then, we sample $L$ instances independently from $\mathcal{G}$ denoted as $G_1, G_2, \ldots, G_L$, and compute projections of each data object $o \in \mathcal{D}$ as follows:
\begin{equation}
    G_i(o) = (h_{i1}(o), h_{i2}(o), \ldots, h_{iK}(o)); i = 1, \ldots, L.
\end{equation}

\noindent \textbf{Indexing points by multi-dimensional indexes.} In each $K$-dimensional projected space, we index points with a multi-dimensional index. The only requirement of the index is that it can  efficiently answer a window query in \cthe{the} low-dimensional space. In this paper, we simply choose \cthe{the} R$^*$-Tree \cite{DBLP:conf/ssd/HwangKCL03} as our index due to an ocean of optimizations and toolboxes, which enables \cthe{the} R$^*$-Tree to perform robustly in practice.  \cthe{The} CR$^*$-Tree \cite{DBLP:conf/sigmod/KimCK01}, X-tree \cite{DBLP:conf/vldb/BerchtoldKK96} or multi-dimensional learned index \cite{DBLP:conf/sigmod/Li0ZY020} can certainly be used to potentially further improve our approach. 

\subsection{Query Phase}
\label{query}
DB-LSH can directly answer an $(r, c)$-NN query with any search radius $r$ by exploiting the $(K, L)$-index that has been built for $(1, c)$-NN in the indexing phase, as described in Section \ref{index}.  Algorithm \ref{alg:rc-nn} outlines the query processing. To find the $(r,c)$-NN of a query $q$, we consider $L$ $K$-dimensional projected spaces in order. For each space, we first compute the hash values of $q$, \ie $G_i(q) = (h_{i1}(q), h_{i2}(q), \ldots, h_{iK}(q))$ (Line \ref{comp_q}). Then, a window query, denoted as  $\WW(G_i(q), w_0 r)$, is conducted using \cthe{the} R$^*$-Tree. To be more specific,  $\WW(G_i(q), w)$ means a query that needs to return points in the following hypercubic region:
\begin{equation}
    [h_{i1}(q) - \frac{w}{2}, h_{i1}(q)  + \frac{w}{2}] \times \cdots \times [h_{iK}(q) - \frac{w}{2}, h_{iK}(q) + \frac{w}{2}].
\end{equation}
Without confusion, we also use $\WW(G_i(q), w)$ to denote a region as above. For each point falling in such a region, we compute its distance to $q$. If the distance is less than $cr$ or we have 
verified $2tL + 1$ points, the algorithm reports the current point and stops. Otherwise, the algorithm returns nothing.
According to Lemma \ref{le:rc-nn}, to be introduced in Section \ref{theory}, DB-LSH is able to correctly answer \cbl{an} $(r, c)$-NN query with a constant success probability. 


\begin{algorithm}[htbp]
\caption{$(r, c)$-NN Query}
\label{alg:rc-nn}
\LinesNumbered 
\KwIn{ $q$: a query point; $r$: query radius; $c$: the approximation ratio; $t$: a positive integer}
\KwOut{A point $o$ or $\varnothing$}
	$cnt \gets 0$\;
	\For{$i = 1$ to $L$}
	{
	    
	    Compute $G_i(q)$\; \label{comp_q}
	    \While{a  point $o \in \WW(G_i(q),w_0\cdot r)$ is found}
	    {
	        $cnt\gets cnt$ + $1$\;
	        \If{$cnt = 2tL + 1$ or $\|q, o\| \leq cr$}
	        {\label{two_case}\Return{$o$}\;}\label{return}
	    }
	}
    \Return{$\varnothing$}\;
\end{algorithm}

\begin{algorithm}[htbp]
\caption{$c$-ANN Query}
\label{alg:c-ann} 
\LinesNumbered 
\KwIn{ $q$: a query point; $c$: the approximation ratio; 
}
\KwOut{A point $o$}
	$r \gets 1$\;
	\While{TRUE}
	{ 
	    $o \gets$ \textbf{call $(r, c)$-NN}\; 
	 \eIf{ $o \neq \varnothing$} 
	{\label{terminate}
	\Return{$o$\;} \label{return2}
	}
	{
	$r \gets cr$\;
	}
	}

\end{algorithm} 

\noindent $\bm{c}$\textbf{-ANN.} A $c$-ANN query can be answered by conducting a series of $(r, c)$-NN queries with $r = 1, c, c^2 ,\ldots$. Algorithm \ref{alg:c-ann} demonstrates the details of finding $c$-ANN. Given a query $q$ and an approximation ratio $c$, the algorithm starts by \cbl{the} $(1, c)$-NN query. After that, if we have found a satisfying object or have accessed enough points \ie $o \neq \varnothing$ (Line \ref{terminate}), the algorithm reports the current 
point and terminates immediately. Otherwise, it 
enlarges the query radius by a factor of $c$ and invokes \cbl{the} $(r, c)$-NN query (Algorithm \ref{alg:rc-nn}) again till the termination conditions are satisfied.  According to Theorem \ref{th:guarantee}, to be introduced in Section \ref{theory}, DB-LSH is able to correctly answer a $c$-ANN query with a constant success probability.

\begin{example}
Figure \ref{fig:cANN} gives an example of answering \cthe{a} $1.5^2$-ANN query by DB-LSH,  where we choose $ K = 2$ and $ L = 1$ for simplicity. Figure \ref{fig:data} and Figure \ref{fig:hash} exhibit the points in the original and projected space, respectively. 
Assume $w_0$ is set to $1.5$. First of all, we issue a $(1,c)$-NN \cbl{query} in the original space (\cbl{the} yellow circle in Figure \ref{fig:data}). To answer this query, we conduct window query $\WW(G(q), w_0)$ in the projected space (\cbl{the} yellow square in Figure \ref{fig:hash}). Since no point is found, \cbl{an}  $(r, c)$-NN query with larger $r$, \ie $r = c$ (\cbl{the} red circle in Figure \ref{fig:data}) is issued, and window query $\WW(G(q), w_0c)$ (\cbl{the} red square in Figure \ref{fig:hash}) is performed accordingly. Then, $o_4$ is found as a candidate and we verify it by computing its original distance to $q$. Since $\|q,o_4\|=2<cr=2.25$ (\cbl{the} blue circle in Figure \ref{fig:data}), $o_4$ is returned as the result.
\begin{figure}[htbp]
	\centering
	{
	\begin{minipage}[c]{0.23\textwidth}
		\subfigure[Original Space]{\includegraphics[width=\linewidth]{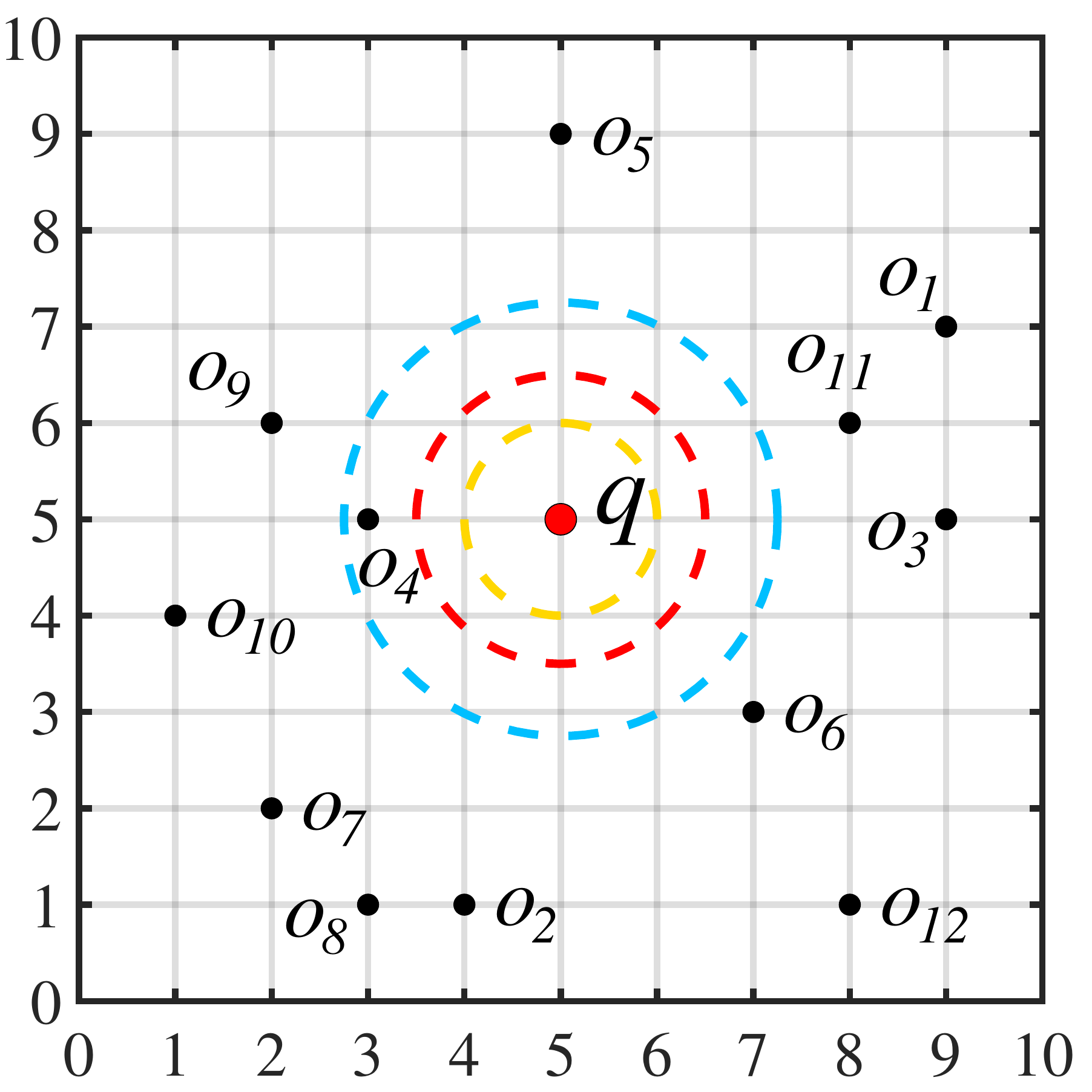}\label{fig:data}}
	\end{minipage}\hspace{.1in}
	\begin{minipage}[c]{0.23\textwidth}
		\subfigure[Projected Space]{\includegraphics[width=\linewidth]{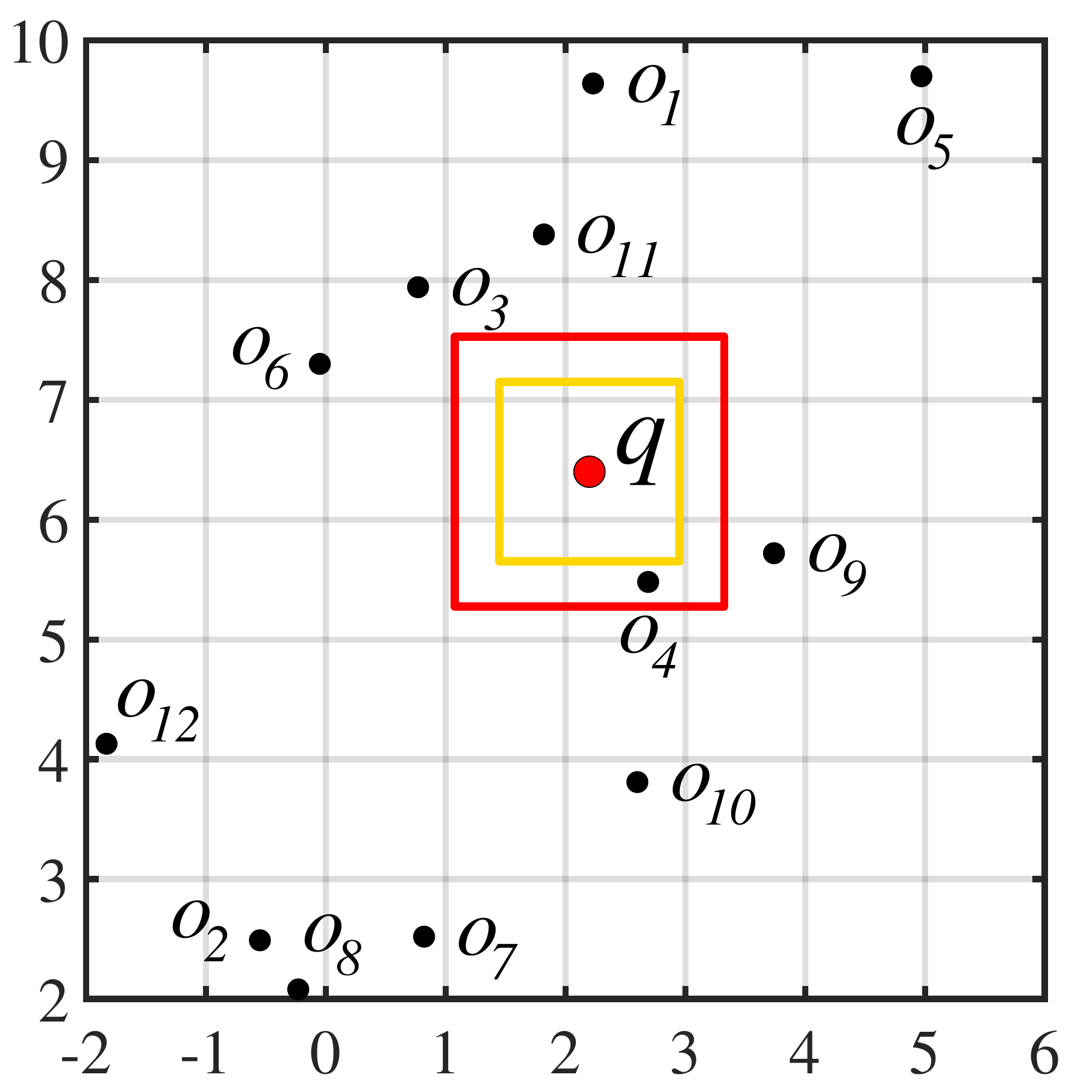}\label{fig:hash}}
	\end{minipage}\hspace{.1in}
	}
	\vspace{-0.1em}\caption{An example of $c$-ANN search using DB-LSH}\vspace{-1em}\label{fig:cANN}
\end{figure}

\end{example}

\noindent 
\textbf{$\bm{(c, k)}$-ANN.}
Algorithm \ref{alg:c-ann} can be easily adapted to answer $(c, k)$-ANN queries. Specifically, it suffices to modify the two termination conditions to the following:
\begin{itemize}
    \item At a certain $(r, c)$-NN query, the total number of objects accessed so far exceeds $2tL + k$ (corresponding to the first case in Line \ref{two_case} of Algorithm \ref{alg:rc-nn}).
    \item At a certain $(r, c)$-NN query, the $k$-th nearest neighbor found so far is within distance $cr$ of $q$ (corresponding to the second case in Line  \ref{two_case} of Algorithm \ref{alg:rc-nn}).
\end{itemize}
DB-LSH terminates if and only if one of the situations happens. Also, apparently Line \ref{return} in \cbl{Algorithm \ref{alg:rc-nn}} (or Line \ref{return2} in Algorithm \ref{alg:c-ann}) should return the $k$ nearest neighbors.

\section{Theoretical Analysis} \label{sec:theory}
\label{theory}
It is essential to provide a theoretical analysis of DB-LSH. First, we discuss the quality guarantees for DB-LSH. Then, we prove that DB-LSH achieves lower query time and space complexities,  with an emphasis on deriving \cbl{a} smaller $\rho^*$.

\subsection{Quality Guarantees}
We demonstrate that DB-LSH is able to correctly answer a $c^2$-ANN query. Before proving it, we first
define two events as follows:

\begin{itemize}[leftmargin=7mm]
    \item[\textbf{E1:}] If there exists a point $o$ satisfying $\|o,q\|\le r$, then $G_i(o)\in \WW(G_i(q),w_0 r)$ for some $i = 1, \ldots, L$;
    \item[\textbf{E2:}] The number of points satisfying two conditions below is no more than $2tL$:  
 1) $\|o,q\|> cr$; and
 2) $G_i(o)\in \WW(G_i(q),w_0 r)$ for some $i = 1, \ldots, L$.
\end{itemize}
\begin{lemma}\label{le:param}
	For given $w_0$ and $t$, by setting $K=\log_{1/p_2} (\frac{n}{t})$ and $L=(\frac{n}{t})^{\rho^*}$ where $p_1=p(1; w_0), p_2=p(c; w_0)$ and $\rho^*=\frac{\ln 1/p_1}{\ln 1/p_2}$, the probability that \emph{\textbf{E1}} occurs is at least $1-1/e$ and the probability that \emph{\textbf{E2}} occurs is at least $1/2$.
\end{lemma}
\begin{proof}
	If there exists a point $o$ satisfying $\|o, q\| \leq r$, then the LSH property implies that for any $h_{ij} \in \mathcal{H}, i = 1, \ldots L, j = 1, \ldots, K$, $\Pr[|h_{ij}(o)-h_{ij}(q)|\le \frac{w_0r}{2}]\ge p(r;w_0r)=p_1$. Then, the probability that $G_i(o)\in \WW(G_i(q),w_0 r)\ge p_1^K$, and thus the probability that \textbf{E1} does not occur will not exceed $(1-p_1^K)^L$. Therefore, $\Pr[\textbf{E1}]\ge 1-(1-p_1^K)^L\ge 1-1/e$ when $K$ and $L$ is set as above.
	Likewise, if there exists a point $o$ satisfying $\|o, q\| > cr$, we have $\Pr[|h_{ij}(o)-h_{ij}(q)|\le \frac{w_0 r}{2}]\le p(cr; w_0r)=p_2$. Then, the probability that $G_i(o)\in \WW(G_i(q),w_0r)\le p^{K}_2=\frac{t}{n}$, and thus the expected number of such points in a certain projected space does not exceed $\frac{t}{n}\cdot n=t$. Therefore, the expected number of such points in all $L$ projected spaces is upper bounded by $tL$. By \textit{ Markov's} \textit{inequality}, we have $\Pr[\textbf{E2}]>1-\frac{tL}{2tL}=1/2$.
\end{proof}

It is easy to see that the probability that \textbf{E1} and \textbf{E2} hold at the same time is a constant, which can be computed as $\Pr[\textbf{E1E2}]=\Pr[\textbf{E1}]-\Pr[\textbf{E1}\overline{\textbf{E2}}]>\Pr[\textbf{E1}]-\Pr[\overline{\textbf{E2}}]=1/2-1/e$. Next, we demonstrate that when \textbf{E1} and \textbf{E2} hold at the same time, Algorithm \ref{alg:rc-nn} is correct for answering an $(r,c)$-NN query.

\begin{lemma}\label{le:rc-nn}
	 Algorithm \ref{alg:rc-nn} answers the $(r,c)$-NN query with at least a constant probability of $1/2-1/e$.
\end{lemma}

\begin{proof}
	Assume that \textbf{E1} and \textbf{E2} hold at the same time, which occurs with at least a constant probability $1/2-1/e$. In this case, if Algorithm \ref{alg:rc-nn} terminates after accessing $2tL+1$ points, then the current point $o$ must satisfy $\|q, o\| \leq cr$ due to \textbf{E2}, and thus a correct result is found. If Algorithm \ref{alg:rc-nn} terminates because of  finding a point satisfying $\|q, o\| \leq cr$, this point is obviously a correct $(r, c)$-NN. If $L$ window queries are over (the algorithm does not terminate because of either already accessing $2tL + 1$ points or  finding a point within $cr$ of $q$), it indicates that no point satisfying $\|q, o\| \leq r$ due to \textbf{E1}. According to the definition of $(r, c)$-NN search, it is reasonable to return nothing. Therefore, when \textbf{E1} and \textbf{E2} hold at the same time, an $(r,c)$-NN query is always correctly answered when Algorithm \ref{alg:rc-nn} terminates. That is,  Algorithm \ref{alg:rc-nn} can answer the $(r,c)$-NN query with at least a constant probability of $1/2-1/e$.
\end{proof}

\begin{theorem}\label{th:guarantee}
	Algorithm \ref{alg:c-ann} returns a $c^2$-ANN $(c>1)$ with at least constant probability of $1/2-1/e$.
\end{theorem}
\begin{proof}
	We show that when \textbf{E1} and \textbf{E2} hold at the same time, Algorithm \ref{alg:c-ann} returns a correct $c^2$-ANN result.  Let $o^*$ be the exact NN of query point $q$ in $\DD$ and   $r^*=\|q, o^*\|$. Without loss of generality, we assume $r^* \geq 1$. Obviously, there must exist a integer $l$ such that $c^l \leq r^* < c^{(l+1)}$. Let $r_0 = c^l$.  When enlarging the search radius $r=1,c,c^2,\dots$, we know that $r$ at termination of Algorithm \ref{alg:c-ann} is at most $cr_0$ due to \textbf{E1}. In this case, according to Lemma \ref{le:rc-nn}, the returned point $o$ satisfies that $\|q, o\|\le c\cdot cr_0\le c^2 r^*$, and thus a correct $c^2$-ANN result. Clearly, if Algorithm 1 stops in a smaller $r$ case for either condition, the returned point satisfy  $\|q, o\| \leq cr < c\cdot cr_0 < c^2r^*$. Therefore, Algorithm \ref{alg:c-ann} returns a $c^2$-ANN $(c>1)$ with at least constant probability of $1/2-1/e$.
\end{proof}

\begin{remark} \label{remark:t}
	Unlike the classic $(K,L)$-index methods,  where $K$ and $L$ are set as $K=\log_{1/p_2} n$ and $L=n^{\rho}$, we introduce a constant $t$ to lessen $K$ and $L$. In this manner, the total space consumption will be greatly reduced. The overhead of this strategy is the need to examine at most $2tL$ candidates instead of $2L$ ones, which seems to cause a higher query cost. However, in fact, 
	none of \cbl{the} efficient LSH methods really build $n^{\rho}$ hash indexes and only check 2 candidates in each index.
	Usually, hash indexes much fewer than $n^{\rho}$ are already able to return a sufficiently accurate $c$-ANN.  Therefore, by introducing $t$, we tend to get $2t$ candidates in one index. This kind of parameter setting is more reasonable and practical.
\end{remark}

\subsection{The bound of $\rho^*$} 
As proven in \cite{DBLP:conf/compgeom/DatarIIM04}, $\rho$ is strictly bounded by $1/c$ when $w_0$ is large enough. Such a large bucket size can not be used experimentally since it implies a very large value of $K$ to effectively differentiate points based on their distance. In contrast, we find that $\rho^*$ has a smaller bound than $1/c$ that can be taken even when the bucket width is not too large. To make a better understanding and simplify the proof, we prove the bound of $\rho^*$ in a special case where $w_0$ is set as $2\gamma c^2$, where $\gamma>0$.

\begin{lemma} \label{le:rho}
	By setting $w_0=2 \gamma c^2,\gamma>0$, $\rho^*$ can be bounded by $1/c^\alpha$, where $\alpha=\frac{\gamma\cdot f(\gamma)}{\int_{\gamma}^{+\infty} f(x) dx}$ and $f(x)$ is the pdf of \cbl{the} standard normal distribution.
\end{lemma}
\begin{proof}
	Recall that $\rho^*=\frac{\ln 1/p_1}{\ln 1/p_2}$, we have
	\begin{equation}
		\rho^*\le\frac{1-p_1}{1-p_2}=\frac{\int_{\gamma c^2}^{+\infty} f(x) dx}{\int_{\gamma c}^{+\infty} f(x) dx},
	\end{equation}
	 according to Lemma 1 in \cite{DBLP:conf/compgeom/DatarIIM04}. Given a $\gamma$, we prove 
	 $	 	\frac{\int_{\gamma c^2}^{+\infty} f(x) dx}{\int_{\gamma c}^{+\infty} f(x) dx}\le 1/{c^\alpha}	 $
	 holds for any $c>1$, which is equivalent to prove the following inequality:
	 \begin{equation}\label{eq:c2_c}
	 	(c^2)^\alpha \int_{\gamma c^2}^{+\infty} f(x) dx \le c^\alpha \int_{\gamma c}^{+\infty} f(x) dx.
	 \end{equation}
	 Define a function $\varphi(u)=u^\alpha \int_{\gamma u}^{+\infty} f(x), u>1$. Inequality \ref{eq:c2_c} holds when $\varphi(u)$ decreases monotonically with $u$.
	 To ensure this, let $\varphi'(u)<0$, where $\varphi'(u)$ is the derivative function of $\varphi(u)$, then we have 
	 $\alpha<\frac{\gamma u\cdot f(\gamma u)}{\int_{\gamma u}^{+\infty} f(x) dx}$. That is to say, inequality \ref{eq:c2_c} holds when $\alpha<\frac{\gamma u\cdot f(\gamma u)}{\int_{\gamma u}^{+\infty} f(x) dx}$. 
	 Denote $\xi(v)=\frac{v\cdot f(v)}{\int_{v}^{+\infty} f(x) dx}$, it can be proven that $\xi(v)$ increases monotonically with $v$ when $v>0$. Since $\gamma>0$ and $u>1$, we have $\gamma u>\gamma$, and thus $\xi(\gamma u) = \frac{\gamma u\cdot f(\gamma u)}{\int_{\gamma u}^{+\infty} f(x) dx}$ is greater than $\xi(\gamma)=\frac{\gamma\cdot f(\gamma)}{\int_{\gamma}^{+\infty} f(x) dx}$. Therefore, $\alpha$ can be set as $\xi(\gamma)$ and then $\rho^*$ is always bounded by $1/{c^\alpha}$ when $w_0=2\gamma c^2$. 
\end{proof}
$\xi(\gamma)>1$ holds when $\gamma>0.7518$, which subsequently provides $\rho^*$ a bound smaller than $1/c$.
The value of $\alpha$ increases with $w_0$, and $\rho^*$ approaches to $0$ when $w_0$ approaches to infinity. That is, the query cost can be very small when $w_0$ is large enough.
However, a large bucket size implies a very large $K$ in order to reduce the number of false positives, so $w_0$ should typically be set to a similar interval range as in other $(K,L)$-index methods. 
Recall that LSB \cite{DBLP:conf/sigmod/TaoYSK09} sets the bucket size to $16$ with approximate ratio $c=2$, we can equivalently set $\gamma = 2$ (\ie $w_0= 4c^2$) to make $w_0$ also be $16$ when $c=2$. Then, according to Lemma \ref{le:rho}, $ \alpha=4.746$ \cbl{and the bound is $1/c^{4.746}$} as compared to the bound of $1/c$ in \cite{DBLP:conf/sigmod/TaoYSK09}. 
\cbl{Note that $\alpha$ can be less than $1$ when $\gamma<0.7518$. In this case, $1/c^\alpha$ no longer seems to be a better bound than $1/c$. However, it will not necessarily lead to $\rho^* > \rho$. Figure \ref{fig:rho_small_w} gives an example that $\rho^* < \rho$ when $\alpha < 1$. By setting $w=0.4c^2$,  $\rho$ exceeds $1/c$ when $c<2$, which means it is not bounded by $1/c$, while $\rho^*$ is always bounded by $1/c^\alpha$ and smaller than $\rho$.  The main reason is that $1/c$ is just an asymptotic bound of $\rho$ approachable only by a very large bucket size, while $1/c^\alpha$ is a non-asymptotic result and $\rho^*$ is always much smaller than $1/c^\alpha$. Besides, it is not necessary to set $\gamma < 0.7518$, since it implies a very large value of $L$. For example, if $w$ is close to $0$, $L$ will be $O(n)$ which makes $(K, L)$-index based methods unpractical. 
Figure \ref{fig:rho_large_w} gives a clear comparison for the decided advantage of $\rho^*$ over $\rho$ by setting  
a reasonable value $w=4c^2$. $\rho$ is very close to $1/c$, while $\rho^*$ has a much smaller bound and decreases rapidly to $0$.}

\begin{figure}[htbp]
    \centering
    \begin{minipage}[c]{0.8\linewidth}
		\centering
		\includegraphics[width=.7\textwidth]{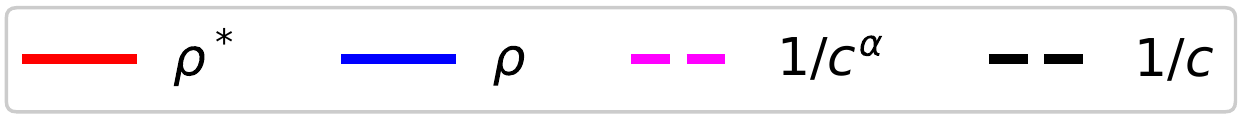}
	\end{minipage}
	\subfigure[$w=0.4c^2$]{
		\begin{minipage}[c]{0.45\linewidth}
			\centering
			\includegraphics[width=1\textwidth]{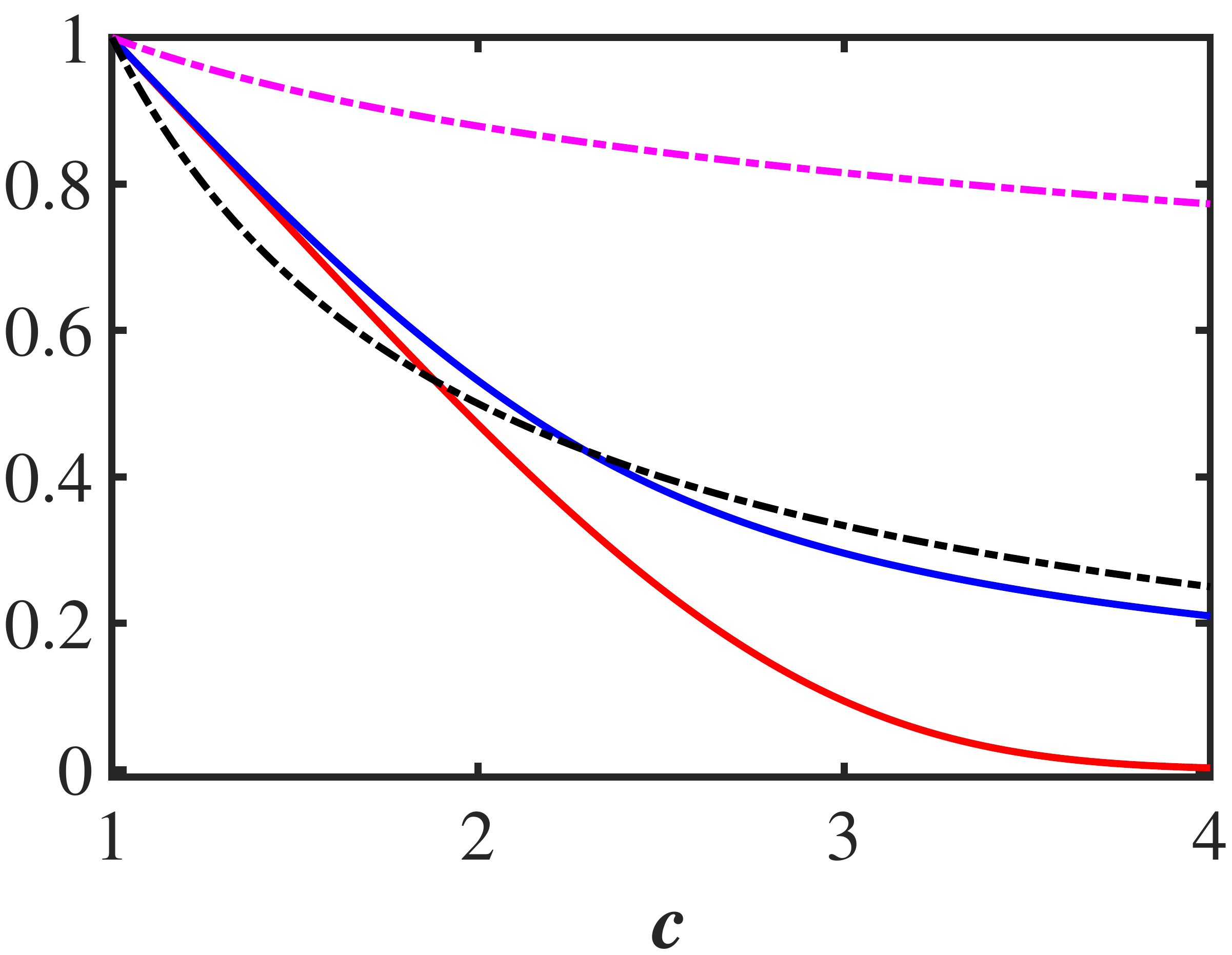}
			\label{fig:rho_small_w}
		\end{minipage}
	}
	\subfigure[$w=4c^2$]{
		\begin{minipage}[c]{0.45\linewidth}
			\centering
			\includegraphics[width=1\textwidth]{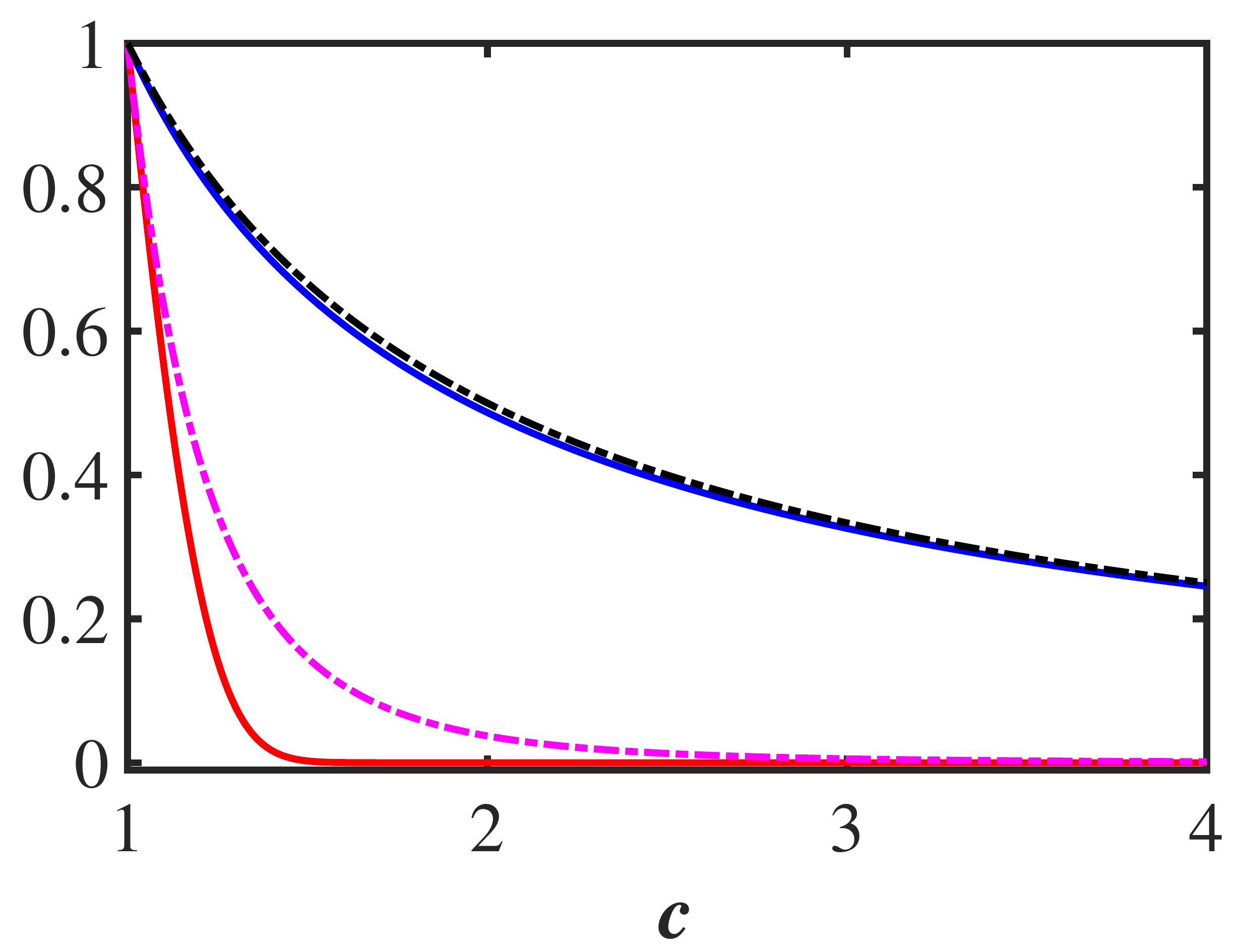}
			\label{fig:rho_large_w}
		\end{minipage}
	}
	\caption{$\rho^*$ v.s. $\rho$}\vspace{-0.5em}
\end{figure}

\subsection{Complexity Analysis} 
Similar to other $(K,L)$-index based methods whose time complexity and space complexity are affected by $\rho$, the complexities of DB-LSH are affected by $\rho^*$. 

\begin{theorem}
	DB-LSH answers a $c^2$-ANN query in $O(n^{\rho^*} d\log n)$ time and $O(n^{1+\rho^*}\log n)$ index size, where $\rho^*$ is bounded by $1/c^{\alpha}$ and smaller than $\rho$ defined in static $(K,L)$-index methods.
\end{theorem}

\begin{proof}
It is obvious that $K=O(\log n)$ and $L=O(n^{\rho^*})$. Therefore, the index size is $O(KL\cdot n)=O(n^{1+\rho^*}\log n)$. In DB-LSH, we need to first compute $K \times L$ hash values of query point, the computational cost of which is $O(KL\cdot d)=O(n^{\rho^*} d\log n)$. When finding candidates, it takes $O(\log n)$ time to find a candidate using R$^*$-Trees. Since we need to retrieve at most $2tL$ candidate points, the cost of generating candidates is $O(\log n \cdot 2tL)=O(n^{\rho^*}\log n)$. In the verification phase, each candidate point spends $O(d)$ time on distance computation, so the total verification cost is $O(2tL\cdot d)=O(n^{\rho^*}d)$. Therefore, the query time of DB-LSH is bounded by $O(n^{\rho^*} d\log n)+O(n^{\rho^*}\log n)+O(n^{\rho^*}d)=O(n^{\rho^*} d\log n)$.
\end{proof}

\section{Experimental study}\label{sec:experiment}

We implement DB-LSH\footnote{https://github.com/Jacyhust/DB-LSH} and the competitors in C++ compiled in a single thread with g++ using O3 optimization. 
All experiments are conducted on a server running 64-bit Ubuntu 20.04 with 2 Intel(R) Xeon(R) Gold 5218 CPUs @ 2.30GHz and 254 GB RAM.

\subsection{Experimental Settings}
\noindent \textbf{Datasets and Queries.} We employ 10 real-world datasets varying in cardinality, dimensionality and types,
which are used widely in existing LSH work \cite{DBLP:conf/icde/LuK20,DBLP:conf/sigmod/LeiHKT20,DBLP:conf/sigmod/Li0ZY020,DBLP:journals/pvldb/LuWWK20,DBLP:journals/pvldb/ZhengZWHLJ20}. \cbl{For the sake of fairness, we make sure  that each dataset is used by at least one of our competitors.} Table \ref{tab:datasets} summarizes the statistics of the datasets. Note that both SIFT10M and SIFT100M consist of points randomly chosen from SIFT1B dataset\footnote{http://corpus-texmex.irisa.fr/}. 
For queries, we randomly select 100 points as queries and remove them from the datasets.

\begin{table}[t]
\caption{Summary of Datasets}
\label{tab:datasets}
\small
\begin{tabular}{|c|c|c|c|}
\hline
\textbf{Datasets}          &\textbf{Cardinality} & \textbf{Dim.} & \textbf{Types}   \\
\hline
\textbf{Audio}   & 54,387   & 192             & Audio           \\
\hline
\textbf{MNIST}   & 60,000   & 784             & Image           \\
\hline
\textbf{Cifar}   & 60,000   & 1024             & Image           \\
\hline
\textbf{Trevi}   & 101,120   & 4096            & Image           \\
\hline
\textbf{NUS}   & 269,648   & 500             & SIFT Description   \\
\hline
\textbf{Deep1M}   & 1,000,000   & 256            & DEEP Description   \\
\hline
\textbf{Gist}   & 1,000,000   & 960             & GIST Description   \\
\hline
\textbf{SIFT10M}   & 10,000,000   & 128       & SIFT Description   \\
\hline
\textbf{TinyImages80M}   & 79,302,017   & 384   & GIST Description   \\
\hline
\textbf{SIFT100M}   & 100,000,000   & 128       & SIFT Description   \\
\hline
\end{tabular}\vspace{-1em}
\end{table}

\noindent \textbf{Competitors.} \cbl{We compare DB-LSH with 5 LSH methods} as mentioned in Section \ref{relatedwork}, \ie LCCS-LSH \cite{DBLP:conf/sigmod/LeiHKT20}, PM-LSH \cite{DBLP:journals/pvldb/ZhengZWHLJ20}, VHP \cite{DBLP:journals/pvldb/LuWWK20} and R2LSH \cite{DBLP:conf/icde/LuK20} and \cbl{LSB-Forest \cite{DBLP:conf/sigmod/TaoYSK09}}. LCCS-LSH \cbl{adopts a query-oblivious LSH indexing strategy with a novel search framework}. PM-LSH is a typical dynamic MQ method that adopts PM-Tree to index the projected data. R2LSH and VHP are representative C2 methods that improve QALSH from the perspective of search regions. \cbl{LSB-Forest is a static $(K, L)$-index method that can answer $ c $-ANN queries for any $ c > 1 $ with only one suit of indexes.}
In addition, to study the effectiveness of query-centric dynamic bucketing strategy in DB-LSH, we design a static $(K, L)$-index method called Fixed Bucketing-LSH (FB-LSH) by replacing the dynamic bucketing part in DB-LSH with the fixed bucketing. Note that FB-LSH is not equivalent to E2LSH since only one suit of $(K,L)$-index is used.

\begin{table*}[]
\centering
\caption{Performance Overview}
\label{tab:overview}
\setlength{\tabcolsep}{3.5mm}{
\begin{tabular}{|cc|c|c|c|c|c|c|c|}
\hline
\multicolumn{2}{|c|}{}                                                    & \textbf{DB-LSH} & \textbf{FB-LSH} & \textbf{LCCS-LSH} & \textbf{PM-LSH} & \textbf{R2LSH} & \textbf{VHP} &  \textbf{LSB-Forest}   \\ \hline
\multicolumn{1}{|c|}{}                                & Query Time (ms)   & \textbf{4.962}  & 5.434           & 5.797             & 5.459           & 8.748          & 11.32                               & 18.52            \\ \cline{2-9} 
\multicolumn{1}{|c|}{}                                & Overall Ratio     & \textbf{1.003}  & 1.008           & 1.006             & 1.003           & 1.005          & 1.006                               & 1.005            \\ \cline{2-9} 
\multicolumn{1}{|c|}{}                                & Recall            & \textbf{0.9268} & 0.8512          & 82.04             & 0.9212          & 0.868          & 0.8580                              & 0.4676           \\ \cline{2-9} 
\multicolumn{1}{|c|}{\multirow{-4}{*}{\textbf{Audio}}}         & Indexing Time (s) & \textbf{0.099}  & 0.164           & 2.126             & 0.166           & 2.764          & 1.626                               &  19.55            \\ \hline\hline
\multicolumn{1}{|c|}{}                                & Query Time (ms)   & \textbf{7.684}  & 9.304           & 19.89             & 13.87           & 12.95          & 15.37                               & 37.35            \\ \cline{2-9} 
\multicolumn{1}{|c|}{}                                & Overall Ratio     & \textbf{1.005}  & 1.018           & 1.007             & 1.005           & 1.005          & 1.008                               &1.010            \\ \cline{2-9} 
\multicolumn{1}{|c|}{}                                & Recall            & \textbf{0.9130} & 0.7580          & 0.8038            & 0.9098          & 0.8756         & 0.8426                              &0.3734           \\ \cline{2-9} 
\multicolumn{1}{|c|}{\multirow{-4}{*}{\textbf{MNIST}}}         & Indexing Time (s) & \textbf{0.149}  & 0.192           & 1.942             & 0.189           & 6.231          & 5.457                               & 92.26            \\ \hline\hline
\multicolumn{1}{|c|}{}                                & Query Time (ms)   & \textbf{12.54}  & 16.37           & 17.66             & 17.53           & 21.81          & 19.31                               & 59.66           \\ \cline{2-9} 
\multicolumn{1}{|c|}{}                                & Overall Ratio     & \textbf{1.002}  & 1.006           & 1.006             & 1.004           & 1.003          & 1.014                               &  1.010            \\ \cline{2-9} 
\multicolumn{1}{|c|}{}                                & Recall            & \textbf{0.9156} & 0.8018          & 0.7150            & 0.8742          & 0.8784         & 0.6322                              &  0.1496           \\ \cline{2-9} 
\multicolumn{1}{|c|}{\multirow{-4}{*}{\textbf{Cifar}}}         & Indexing Time (s) & \textbf{0.149}  & 0.209           & 1.941             & 0.199           & 8.261          & 6.844                               &  146.27           \\ \hline
\multicolumn{1}{|c|}{}                                & Query Time (ms)   & \textbf{48.20}  & 61.74           & 113.7             & 52.23           & 53.10          & 176.47                              & 271.56           \\ \cline{2-9} 
\multicolumn{1}{|c|}{}                                & Overall Ratio     & \textbf{1.001}  & 1.010           & 1.003             & 1.002           & 1.003          & 1.003                               &  1.007            \\ \cline{2-9} 
\multicolumn{1}{|c|}{}                                & Recall            & \textbf{0.9338} & 0.6818          & 0.7816            & 0.8918          & 0.8100         & 0.8798                              &  0.1588           \\ \cline{2-9} 
\multicolumn{1}{|c|}{\multirow{-4}{*}{\textbf{Trevi}}}         & Indexing Time (s) & \textbf{0.232}  & 0.374           & 6.572             & 0.386           & 46.08          & 44.05                               &  1347.9           \\ \hline\hline
\multicolumn{1}{|c|}{}                                & Query Time (ms)   & \textbf{36.07}  & 58.75           & 79.15             & 68.38           & 93.13          & 103.33             &  155.72           \\ \cline{2-9} 
\multicolumn{1}{|c|}{}                                & Overall Ratio     & \textbf{1.0008} & 1.011           & 1.004             & 1.011           & 1.012          & 1.010                               &  1.009            \\ \cline{2-9} 
\multicolumn{1}{|c|}{}                                & Recall            & \textbf{0.5532} & 0.4656          & 0.5376            & 0.4637          & 0.4494         & 0.4972                              &  0.1080           \\ \cline{2-9} 
\multicolumn{1}{|c|}{\multirow{-4}{*}{\textbf{NUS}}}           & Indexing Time (s) & \textbf{0.768}  & 1.655           & 40.032            & 1.190           & 23.40          & 15.86                               &  798.45           \\ \hline\hline
\multicolumn{1}{|c|}{}                                & Query Time (ms)   & \textbf{127.16} & 170.24          & 163.24            & 327.58          & 188.84         & 243.53                              &  377.60           \\ \cline{2-9} 
\multicolumn{1}{|c|}{}                                & Overall Ratio     & \textbf{1.004}           & 1.010           & 1.004             & 1.004           & 1.005          & 1.014                               &  1.003   \\ \cline{2-9} 
\multicolumn{1}{|c|}{}                                & Recall            & \textbf{0.8784} & 0.7376          & 0.8530            & 0.8594          & 0.8354         & 0.5048                              & 0.4524           \\ \cline{2-9} 
\multicolumn{1}{|c|}{\multirow{-4}{*}{\textbf{Deep1M}}}        & Indexing Time (s) & \textbf{5.704}  & 7.856           & 159.41            & 6.141           & 61.79          & 34.57                               &  3498.3           \\ \hline\hline
\multicolumn{1}{|c|}{}                                & Query Time (ms)   & \textbf{164.03} & 265.90          & 335.67            & 339.63          & 288.63         & 384.77                              &  761.02           \\ \cline{2-9} 
\multicolumn{1}{|c|}{}                                & Overall Ratio     & \textbf{1.004}  & 1.007           & 1.003             & 1.006           & 1.010          & 1.016                               &  1.005           \\ \cline{2-9} 
\multicolumn{1}{|c|}{}                                & Recall            & \textbf{0.8098} & 0.7360          & 0.7248            & 0.7566          & 0.6442         & 0.5180                              &  0.2736           \\ \cline{2-9} 
\multicolumn{1}{|c|}{\multirow{-4}{*}{\textbf{Gist}}}          & Indexing Time (s) & \textbf{6.056}  & 7.811           & 178.74            & 8.038           & 139.93         & 105.98                              & 11907            \\ \hline\hline
\multicolumn{1}{|c|}{}                                & Query Time (ms)   & \textbf{963.17} & 2633.9          & 2774.66           & 1922.4          & 3998           & 9723.4                              & 2667.9           \\ \cline{2-9} 
\multicolumn{1}{|c|}{}                                & Overall Ratio     & \textbf{1.001}  & 1.002           & 1.002             & 1.001           & 1.001          & 1.006                               & 1.001            \\ \cline{2-9} 
\multicolumn{1}{|c|}{}                                & Recall            & \textbf{0.9602} & 0.9420          & 0.9192            & 0.9469          & 0.9560         & 0.8248                              &  0.7206           \\ \cline{2-9} 
\multicolumn{1}{|c|}{\multirow{-4}{*}{\textbf{SIFT10M}}}       & Indexing Time (s) & \textbf{86.49}  & 123.46          & 159.31            & 101.71          & 506.13         & 263.19                              &  23631            \\ \hline\hline
\multicolumn{1}{|c|}{}                                & Query Time (ms)   & \textbf{14511}  & 28854           & 21101             & 29023           & 35396          & 164194                              &  \textbackslash{} \\ \cline{2-9} 
\multicolumn{1}{|c|}{}                                & Overall Ratio     & \textbf{1.002}  & 1.004           & 1.002             & 1.005           & 1.035          & 1.014                               & \textbackslash{} \\ \cline{2-9} 
\multicolumn{1}{|c|}{}                                & Recall            & \textbf{0.8922} & 0.8144          & 0.8384            & 0.8164          & 0.6303         & 0.7720                              &  \textbackslash{} \\ \cline{2-9} 
\multicolumn{1}{|c|}{\multirow{-4}{*}{\textbf{TinyImages80M}}} & Indexing Time (s) & \textbf{1198.9} & 2663.3          & 23911             & 2153.5          & 6508.1         & 4265.1                              &  \textbackslash{} \\ \hline\hline
\multicolumn{1}{|c|}{}                                & Query Time (ms)   & \textbf{7961.6} & 10287           & 25342             & 26724           & 25467          & 163531                              &  \textbackslash{} \\ \cline{2-9} 
\multicolumn{1}{|c|}{}                                & Overall Ratio     & \textbf{1.001}  & 1.009           & 1.004             & 1.001           & 1.019          & 1.006                               &  \textbackslash{} \\ \cline{2-9} 
\multicolumn{1}{|c|}{}                                & Recall            & \textbf{0.9618} & 0.7960          & 0.8568            & 0.9597          & 0.6180         & 0.7980                              &  \textbackslash{} \\ \cline{2-9} 
\multicolumn{1}{|c|}{\multirow{-4}{*}{\textbf{SIFT100M}}}      & Indexing Time (s) & \textbf{1638.1} & 3414.3          & 10912             & 2552.6          & 5404.6         & 3442.9                              &  \textbackslash{} \\ \hline
\end{tabular}}
\vspace{-0.5em}
\end{table*}

\noindent \textbf{Parameter Settings.}  By default, all algorithms are conducted to answer $(c,k)$-ANN queries with $k=50$. 
For DB-LSH, we set the approximation ratio $c=1.5$ and $w=4c^2$. $L$ is fixed as $5$. $K=12$ for the datasets with cardinality greater than 1M and $K=10$ for the rest datasets.
\cbl{Parameter settings of competitors follow the original papers or their source codes. 
Specifically,}
for LCSS-LSH, we set $m=64$ and $\# probes\in\{256,512\}$.
For PM-LSH, we set $c=1.5$ and use $m=15$ hash functions, $\beta=0.08$.
For R2LSH, we are recommended to set $\lambda$, $m$ and $\beta$ to $0.7$, $40$ and $30$.
For VHP, we set $t_0=1.4$ and $m=60$ for the datasets except Gist, Trevi and Cifar. For these three datasets, $m$ is set as $80$  since they have much higher dimensionality.
\cbl{For LSB-Forest, we set $B=1024\sim4096\text{KB}$ based on the dimensionality of the datasets. Then $l$ and $m$ can be computed by $l=\sqrt{dn/B}$ and $m=\log_{1/p_2}dn/B$.}
\cbl{To achieve comparable query accuracy with the competitors, we increase the total number of leaf entries in LSB-Forest from $4Bl/d$ to $40Bl/d$.}
For FB-LSH, we set the approximation ratio $c=1.5$ and $w=4c^2$. $K$ is fixed as $5$ and $L$ ranges from $10$ to $12$ based on the cardinality of the datasets.
\ifLONGVERSION
\todo{
}
\fi


\noindent \textbf{Evaluation Metric.} There are five metrics in total. Two metrics are used to evaluate
the indexing performance: namely, index size and indexing time. 
Three metrics are used to evaluate the query performance: query time, overall ratio and recall. For a $(c,k)$-ANN query, let the returned set be $R=\{o_1,\dots,o_k\}$ with points sorted in ascending order of  their distances to the query point and the exact $k$-NN  $R^*=\{o^*_1,\dots,o^*_k\}$, then the overall ratio and recall are defined as follows \cite{DBLP:journals/pvldb/ZhengZWHLJ20}.
\begin{equation}
	\mathit{OverallRatio} = \frac{1}{k} \sum_{i=1}^{k} \frac{\|q,o_i\|}{\|q,o_i^*\|}
\end{equation}
\begin{equation}
	\mathit{Recall}=\frac{\lvert R \cap R^* \rvert}{k}
\end{equation}
We repeatedly conduct each algorithm $10$ times for all $100$ queries and report the average query time, overall ratio and recall. Since \cbl{LSB-Forest,} R2LSH and VHP are disk-based methods, we only take their CPU time as the query time for fairness. 
\cbl{For FB-LSH, we omit the search time for candidates in R$^*$-Tree when computing the query time so as to mimic the fast lookup of candidates through hash tables in static $(K,L)$-index methods. Such time cannot be ignored in DB-LSH.}

\subsection{Performance Overview}
 In this subsection, we provide an overview of \cbl{the} average query time, overall ratio, recall and indexing time of all algorithms \cbl{with default parameter settings} on all datasets, as shown in Table \ref{tab:overview}. \cbl{We do not run LSB-Forest on TinyImages80M and SIFT100M, since their storage consumption is considerably huge (more than 10TB to store the indexes).}
\subsubsection{\textbf{DB-LSH and FB-LSH}}
we first make a brief comparison of DB-LSH and FB-LSH, where the number of hash functions $K \times L$ is set to the same value. The only difference between them is whether a query-centric bucket is used or not. As we can see from Table \ref{tab:overview}, 
DB-LSH saves $10$-$70\%$ of the query time compared to FB-LSH but reaches a higher recall and smaller overall ratio. In other words, DB-LSH achieves better accuracy with higher efficiency. The main reason is that although DB-LSH spends more time searching for candidates in the \cbl{R$^*$-Trees}, the number of required candidates is reduced due to the high quality of candidates in query-centric buckets. 

\subsubsection{\textbf{Indexing Performance}} 
    
The indexing time and index size of all algorithms with the default settings are considered in this set of experiments. 
Since the index size of all algorithms \cbl{except LSB-Forest} can be easily estimated by 
$
    \textit{IndexSize} = n\times \#\textit{HashFunctions}
$, we compare the index size by the number of hash functions used in each algorithm \cbl{as mentioned in the parameter settings} and do not list them again in the Table \ref{tab:overview}.  We can see that the index sizes are close for all algorithms except PM-LSH, which demonstrates that DB-LSH eliminates the space consumption issue in $(K,L)$-index methods. \cbl{In LSB-Forest, data points  are also stored in each indexes, which leads to extremely large space consumption. Besides, the value of $L$ in LSB-Forest is $O(\sqrt{n})$. It  also makes LSB-Forest ill-adapted to  the  large-scale datasets. For example, $L$ reaches to $ 485 $ for Gist and $560$ for SIFT10M.}
For the indexing time, as shown in Table \ref{tab:overview}, we have the following observations: 
(1) DB-LSH achieves the smallest indexing time on all datasets. The reason is twofold. 
First, DB-LSH adopts the bulk-loading strategy to construct R$^*$-Trees, which is a more efficient strategy than conventional insertion strategies.
It takes less time to construct $5$ R$^*$-Trees than PM-LSH to build a PM-Tree.
Second, DB-LSH requires only $5$ indexes, which is much smaller than those in LCCS-LSH, R2LSH and VHP. In addition, R2LSH and VHP have close indexing time since they both adopt B$^+$-Trees as indexes. LCSS has a much longer indexing time than other algorithms due to its complex index structure, CSA. \cbl{The indexing time of LSB-Forest is also very long  because LSB-Forest uses several times the number of indexes than other algorithms.
}
(2) \cbl{The} indexing time is almost determined by the cardinality of the dataset and it increases super-linearly with cardinality in all algorithms.
For example, MNIST and Cifar have the same cardinality and almost the same indexing time. All algorithms take more than 10 times longer to build indexes on dataset SIFT100M than on SIFT10M. It implies that it is time-consuming to construct indexes for very large-scale datasets, and therefore, the smallest indexing time gives DB-LSH a great advantage.

\subsubsection{\textbf{Query Performance}}
In this set of experiments, we study the average query time, recall and overall ratio of all algorithms in the default settings. According to the results shown in Table \ref{tab:overview}, we have the following observations:
(1) DB-LSH offers the best query performance on all datasets. 
\cthe{The} higher recall, smaller overall ratio and shorter query time indicate DB-LSH outperforms all competitor algorithms on both efficiency and accuracy. In particular, on very large-scale datasets TinyImages80M and SIFT100M ($14.5$s and $7.9$s), DB-LSH not only takes just about half query time of PM-LSH, R2LSH, VHP and \cbl{LSB-Forest}, but also reaches a higher accuracy. 
Only LCCS-LSH and FB-LSH achieve \cbl{the} comparable query time  on these two large-scale datasets ($21$s and $10.3$s).
\cbl{The reason DB-LSH achieves the best performances can be concluded as follows: a) compared with query-oblivious methods (LCCS-LSH, LSB-Forest), query-centric methods can obtain higher quality candidates since they address the hash boundary issue; b) compared with other query-centric methods (C2), both MQ and DB-LSH perform better due to the bounded search region; c) compared with MQ that adopts only one index, DB-LSH uses $ L $ indexes to miss fewer exact NNs, and thus achieving better recall and ratio.}
(2) The query accuracy, especially recall, varies with datasets.
\cbl{All algorithms can achieve $80$-$90\%$ recall on most datasets. On NUS, all algorithms perform slightly inferior due to intrinsically complex distribution (that can be quantified by \textit{relative contrast} and \textit{local intrinsic dimensionality} \cite{DBLP:journals/pvldb/ZhengZWHLJ20,DBLP:conf/icml/HeKC12,DBLP:journals/tkde/LiZSWLZL20}), but DB-LSH still has a lead.}
(3) The query performance of VHP and R2LSH are considerably worse than other algorithms on large-scale datasets TinyImages80M and SIFT100M. VHP takes as long as linear scan ($164$s and $163$s) and R2LSH is difficult to reach an acceptable recall ($0.63$ and $0.61$) or overall ratio. 
Therefore, we do not report the results of them on TinyImages80M and SIFT100M in the subsequent experiments.
\cbl{(4) No matter which datasets, LSB-Forest always needs the longest query time to reach a similar accuracy. Its query time grows rapidly with the cardinality and dimensionality of the dataset. As many as $ O(\sqrt{nd}) $ index accesses make LSB-Forest not comparable to others, so we do not report it in the rest experiments.}

\subsection{Evaluation of Query Performance}

\subsubsection{\textbf{Effect of $n$}}

In order to investigate how the dataset cardinality affects the query performance, we randomly pick up $0.2n, 0.4n, 0.6n, 0.8n$ and $ n$ data points from the original dataset and compare \cbl{the} query performance of all algorithms on them \cbl{in the default  parameters.} 
Due to the space limitation, we only report the results on Gist and TinyImages80M, which are representative due to their different cardinality and dimensionality. The comparative results are shown in Figure \ref{expe:n_time}-\ref{expe:n_ratio}. 
 \begin{figure}[htbp]
    \begin{minipage}[c]{0.8\linewidth}
		\centering
		\includegraphics[width=.7\textwidth]{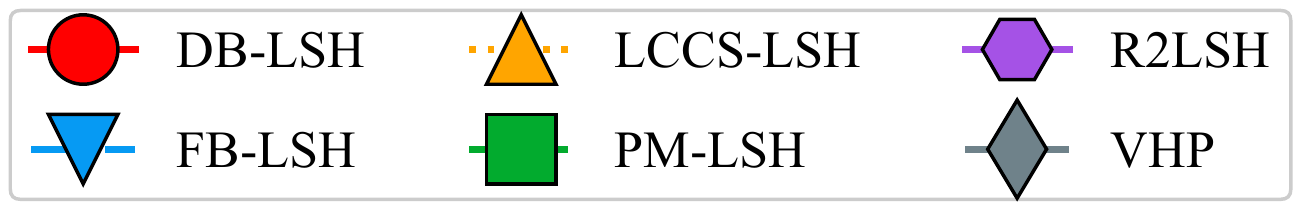}
	\end{minipage}
 	\centering
 	\subfigure[Gist]{
 		\begin{minipage}[c]{0.45\linewidth}
 			\centering
 			\includegraphics[width=1\textwidth]{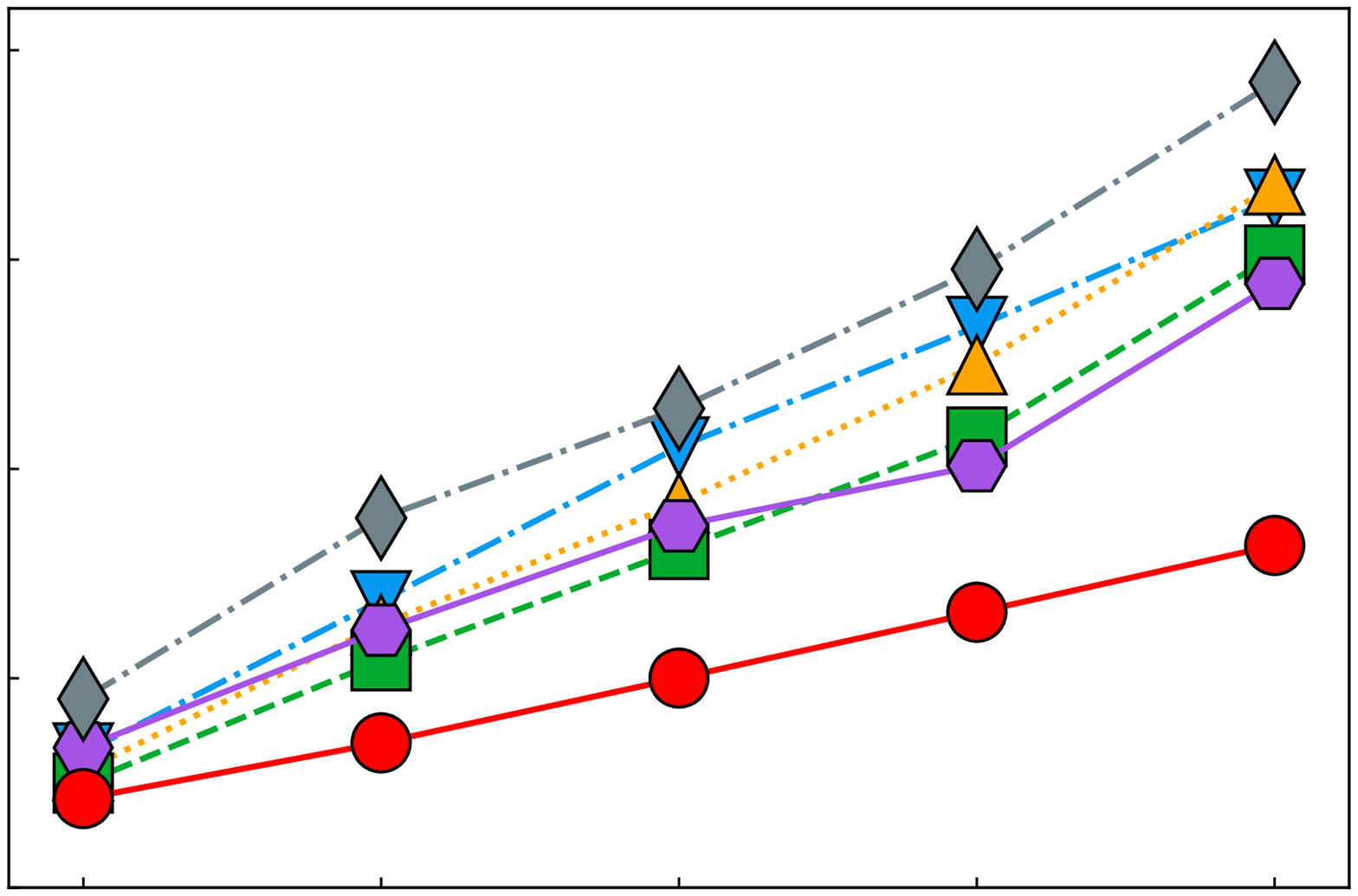}
 			\vspace{-1em}
 		\end{minipage}
 	}
 	\subfigure[TinyImages80M]{
 		\begin{minipage}[c]{0.45\linewidth}
 			\centering
 			\includegraphics[width=1\textwidth]{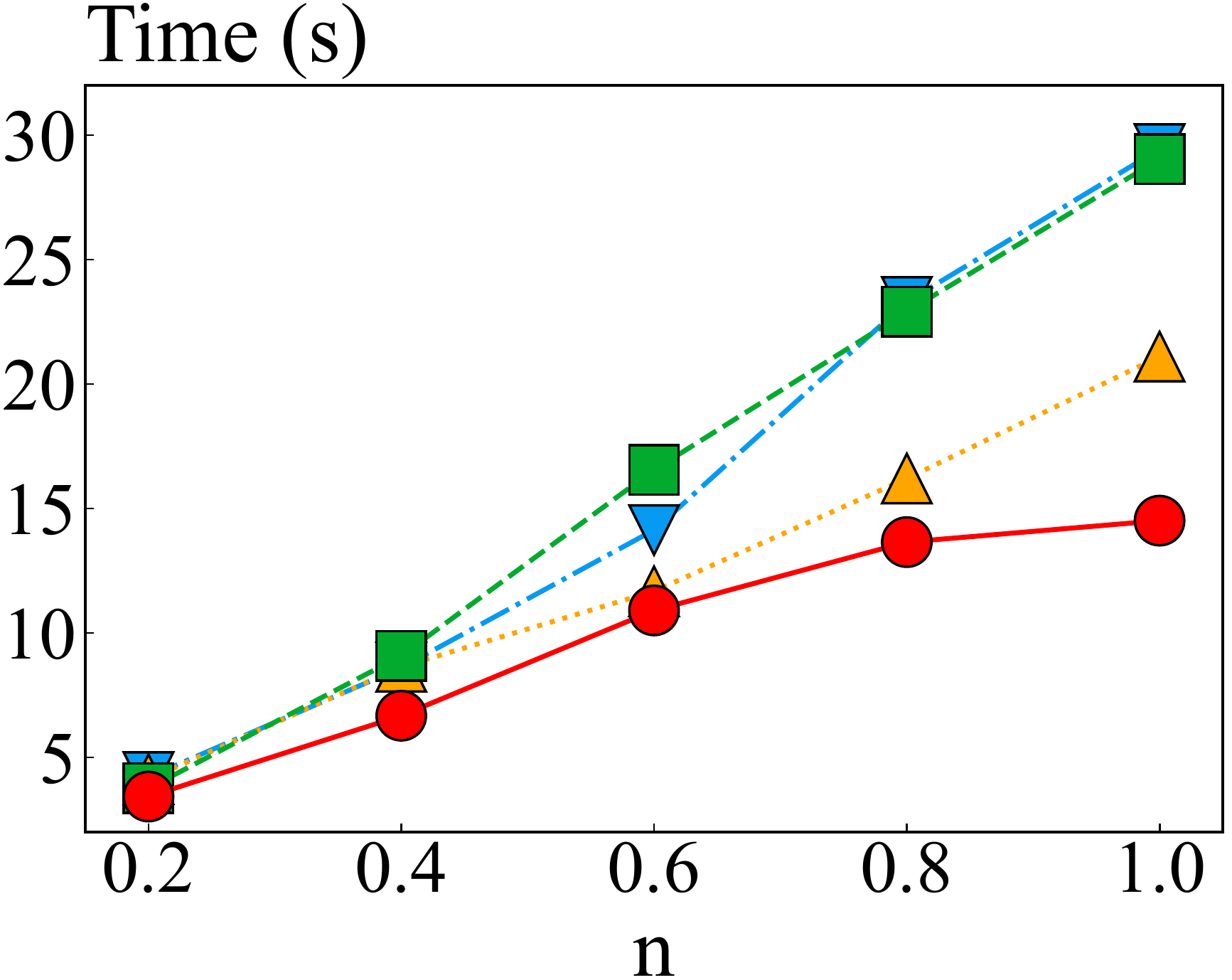}
 			\vspace{-1em}
 		\end{minipage}
 	}
 	\vspace{-0.8em}\caption{Query Time when Varying $n$}\vspace{-0.6em}
 	\label{expe:n_time}
 \end{figure}
 \begin{figure}[htbp]
 	\centering
 	\subfigure[Gist]{
 		\begin{minipage}[c]{0.45\linewidth}
 			\centering
 			\includegraphics[width=1\textwidth]{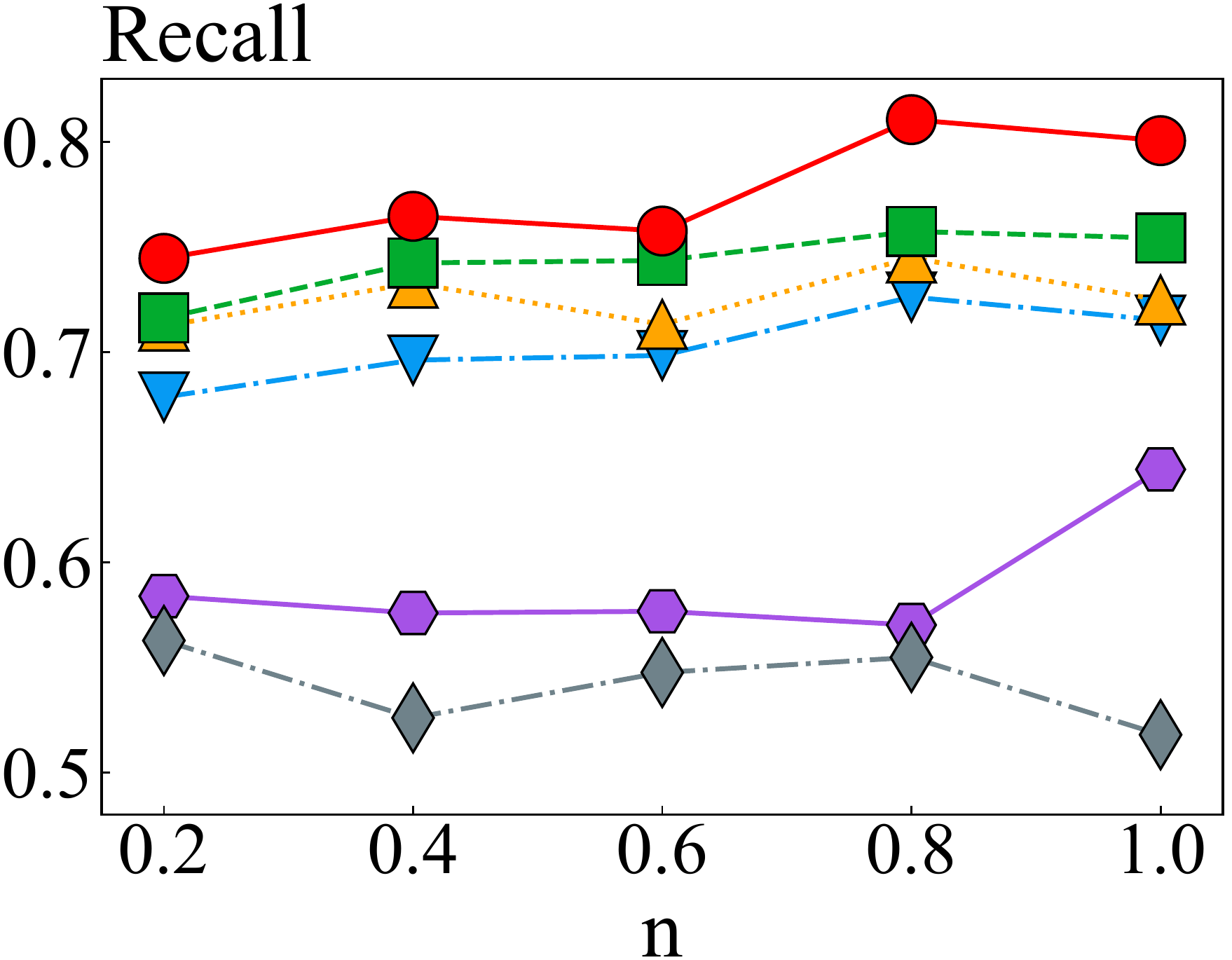}
 			\vspace{-1em}
 		\end{minipage}
 	}
 	\subfigure[TinyImages80M]{
 		\begin{minipage}[c]{0.45\linewidth}
 			\centering
 			\includegraphics[width=1\textwidth]{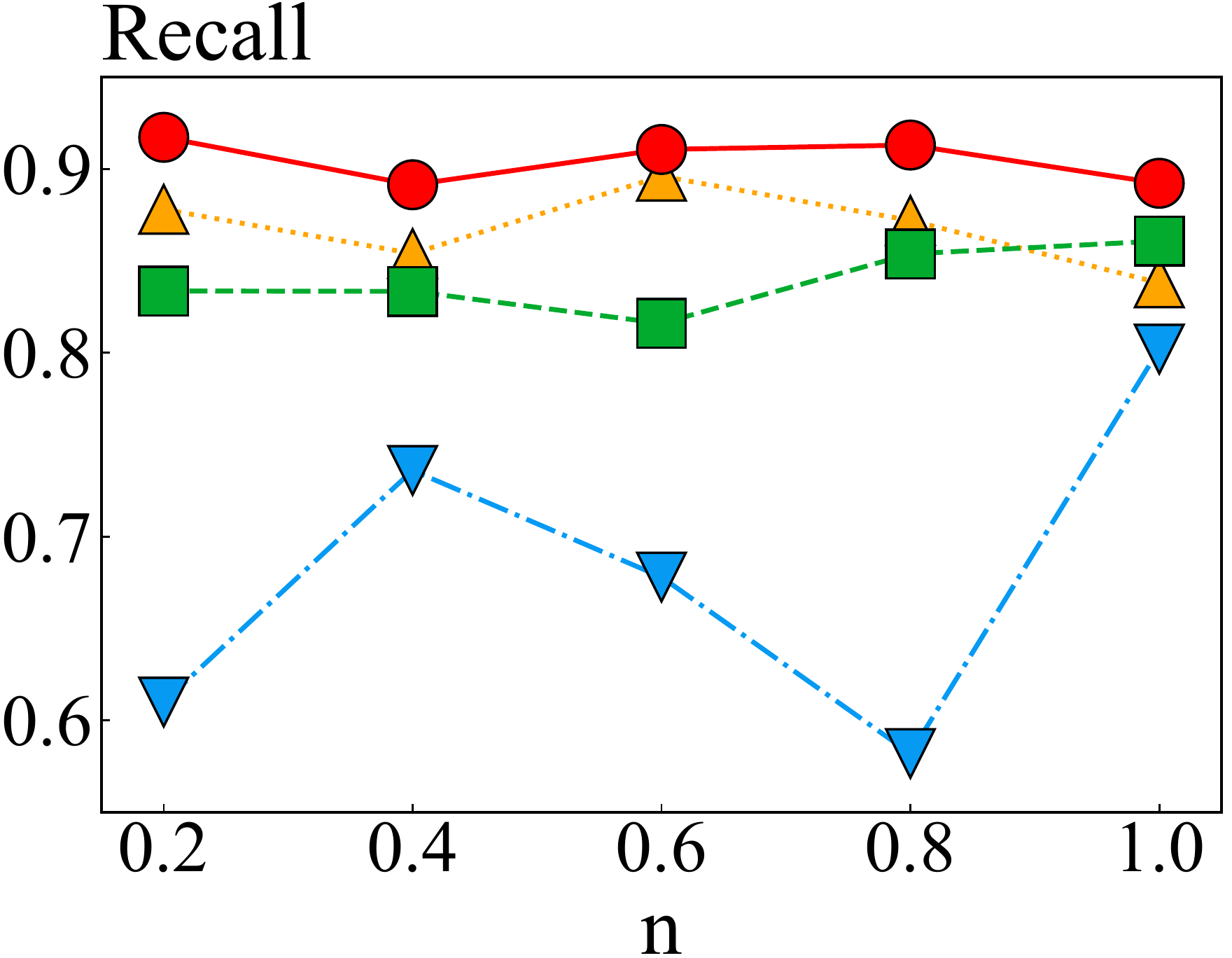}
 			\vspace{-1em}
 		\end{minipage}
 	}
 	\vspace{-0.8em}\caption{Recall when Varying $n$}\vspace{-0.6em}
 	\label{expe:n_recall}
 \end{figure}
 \begin{figure}[!h]
 	\centering
 	\subfigure[Gist]{
 		\begin{minipage}[c]{0.45\linewidth}
 			\centering
 			\includegraphics[width=1\textwidth]{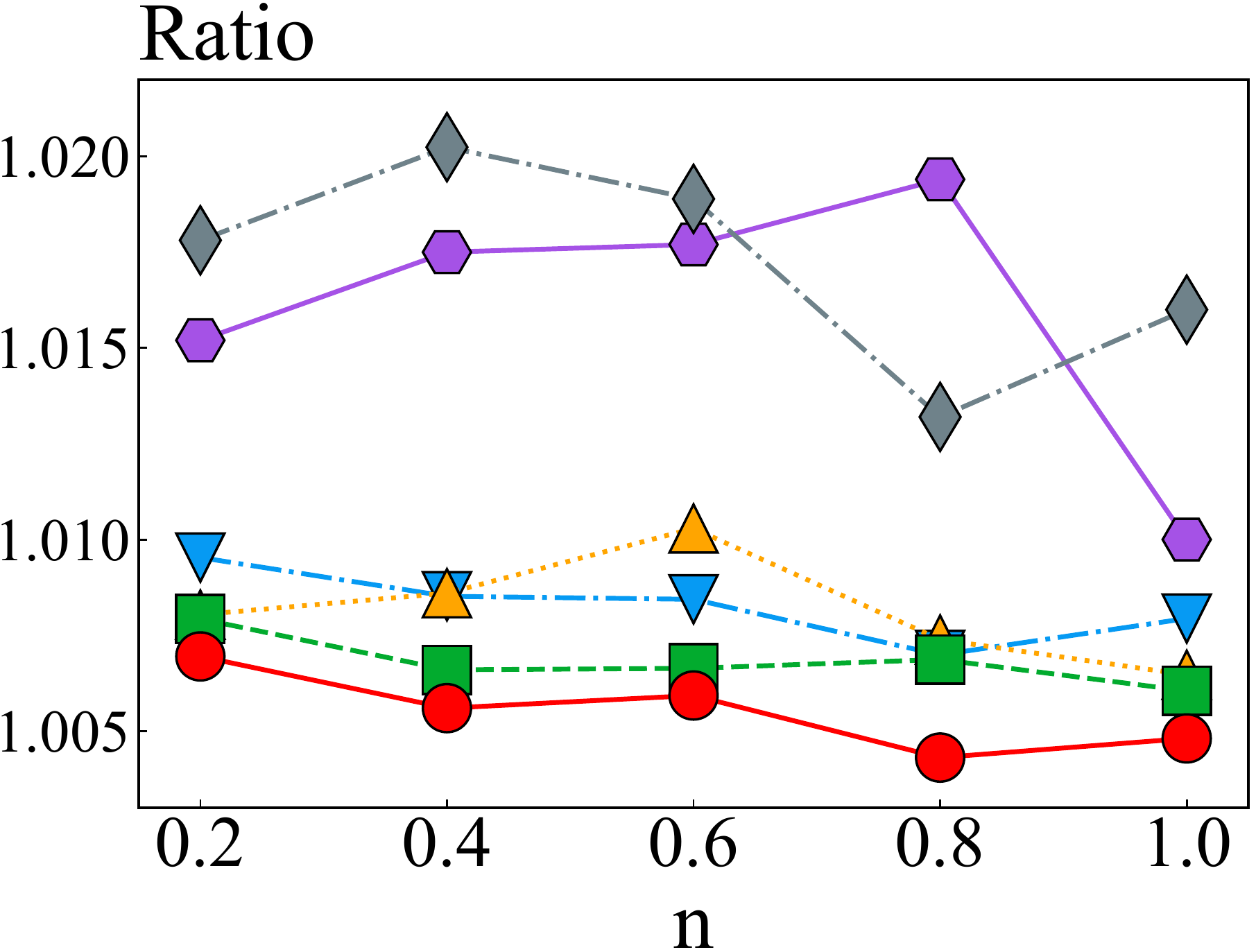}
 			\vspace{-1em}
 		\end{minipage}
 	}
 	\subfigure[TinyImages80M]{
 		\begin{minipage}[c]{0.45\linewidth}
 			\centering
 			\includegraphics[width=1\textwidth]{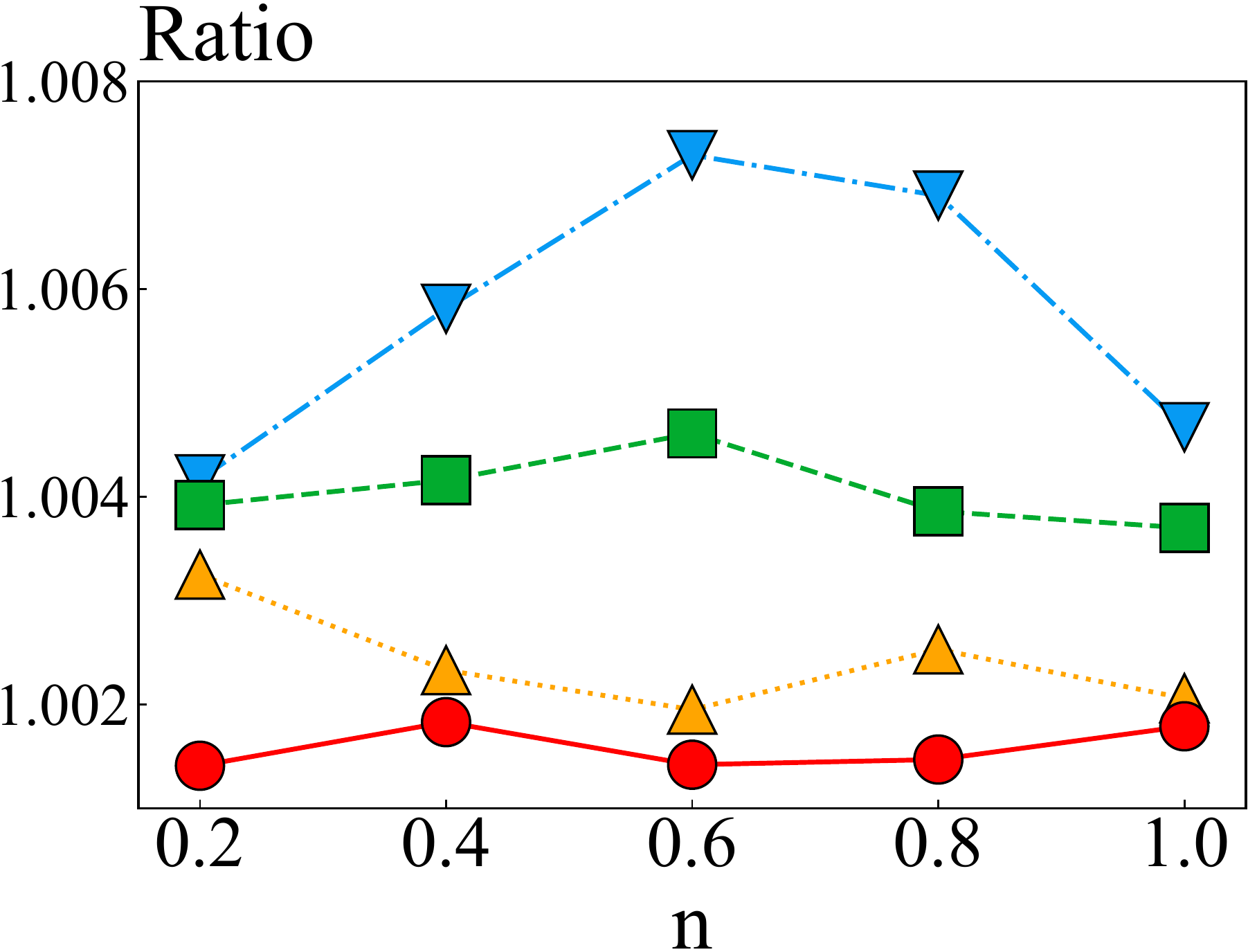}
 			\vspace{-1em}
 		\end{minipage}
 	}
 	\vspace{-0.8em}\caption{Overall Ratio when Varying $n$} \vspace{-0.6em}
 	\label{expe:n_ratio}
 \end{figure}
Clearly, DB-LSH has a lead advantage over all competitors under all evaluation metrics when varying \cthe{the} cardinality. Although the query time increases with \cthe{the} cardinality, DB-LSH grows much slower than other algorithms. The reason is that DB-LSH truly achieves \cthe{a} sub-linear query cost. 
In terms of query accuracy, all algorithms, especially DB-LSH, LCCS-LSH and PM-LSH, achieve relatively stable recall and overall ratio, because query accuracy depends mainly on the data distribution. Although the cardinality increases, the data distribution remains essentially the same, and therefore the accuracy does not change much. The accuracy of FB-LSH is a bit unsteady due to hash boundary issue. As we can see, DB-LSH keeps performing better than all competitor algorithms.

\begin{figure*}[!h]
    \centering
    \begin{minipage}[c]{0.8\linewidth}
		\centering
		\includegraphics[width=.7\textwidth]{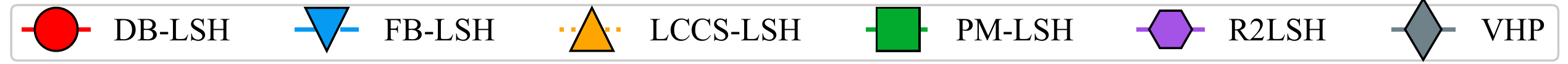}
	\end{minipage}
	\subfigure[Recall on Gist]{
		\begin{minipage}[c]{0.23\linewidth}
			\centering
			\includegraphics[width=1\textwidth]{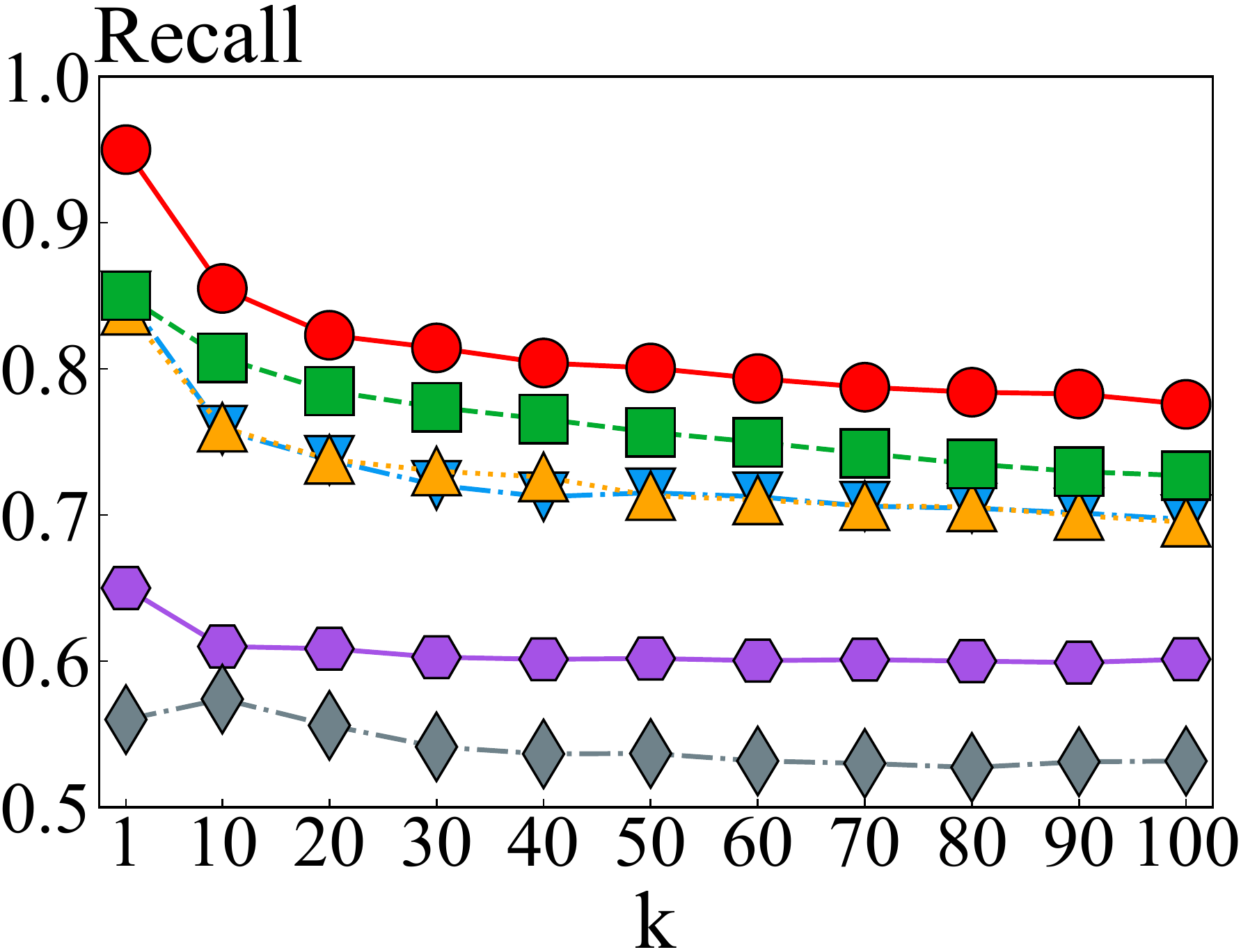}
				\vspace{-1.05em}
		\end{minipage}
	}
	\subfigure[OverRatio on Gist]{
		\begin{minipage}[c]{0.23\linewidth}
			\centering
			\includegraphics[width=1\textwidth]{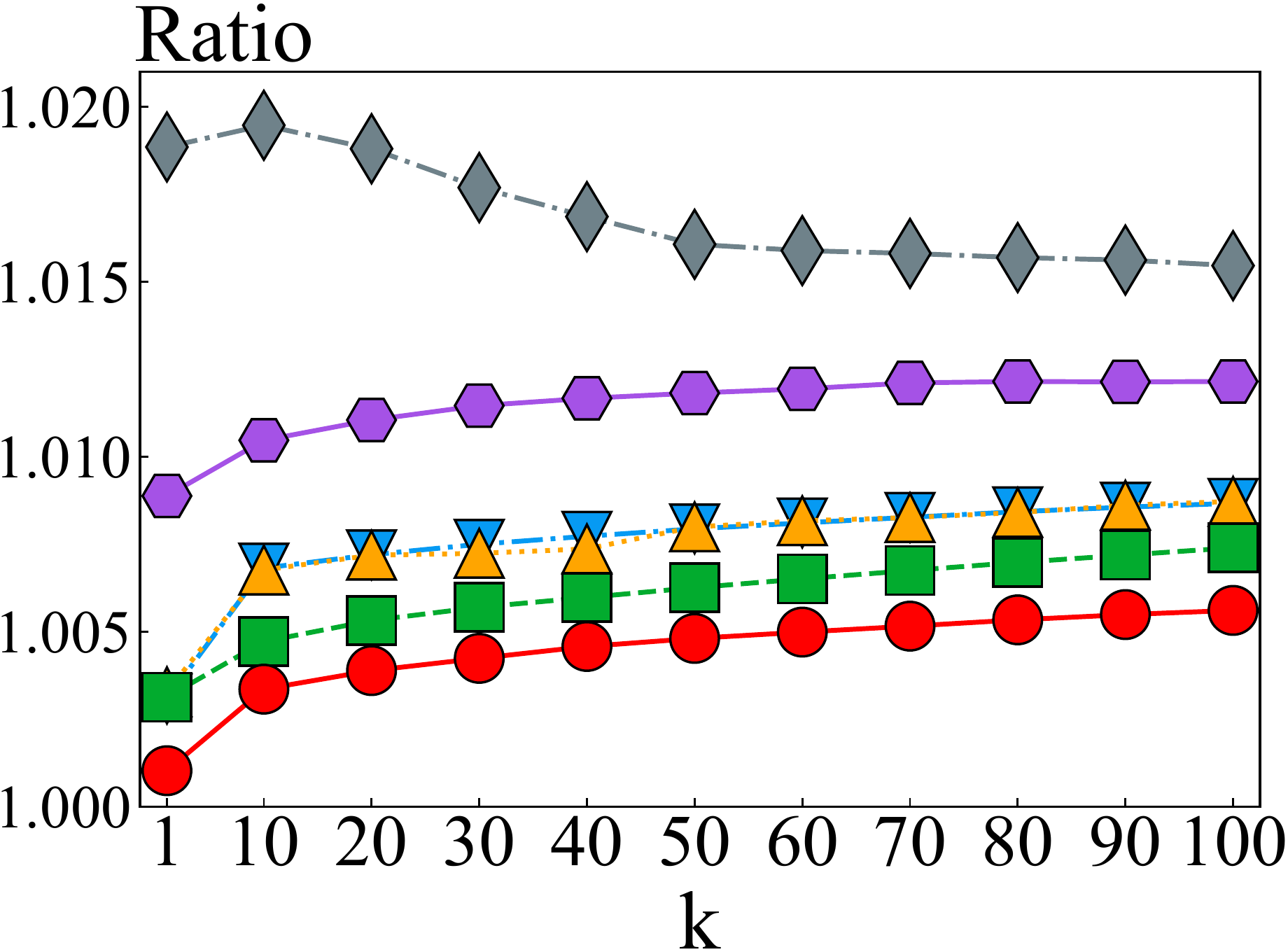}
			\vspace{-1.05em}
		\end{minipage}
	}
	\subfigure[Recall on TinyImages80M]{
		\begin{minipage}[c]{0.23\linewidth}
			\centering
			\includegraphics[width=1\textwidth]{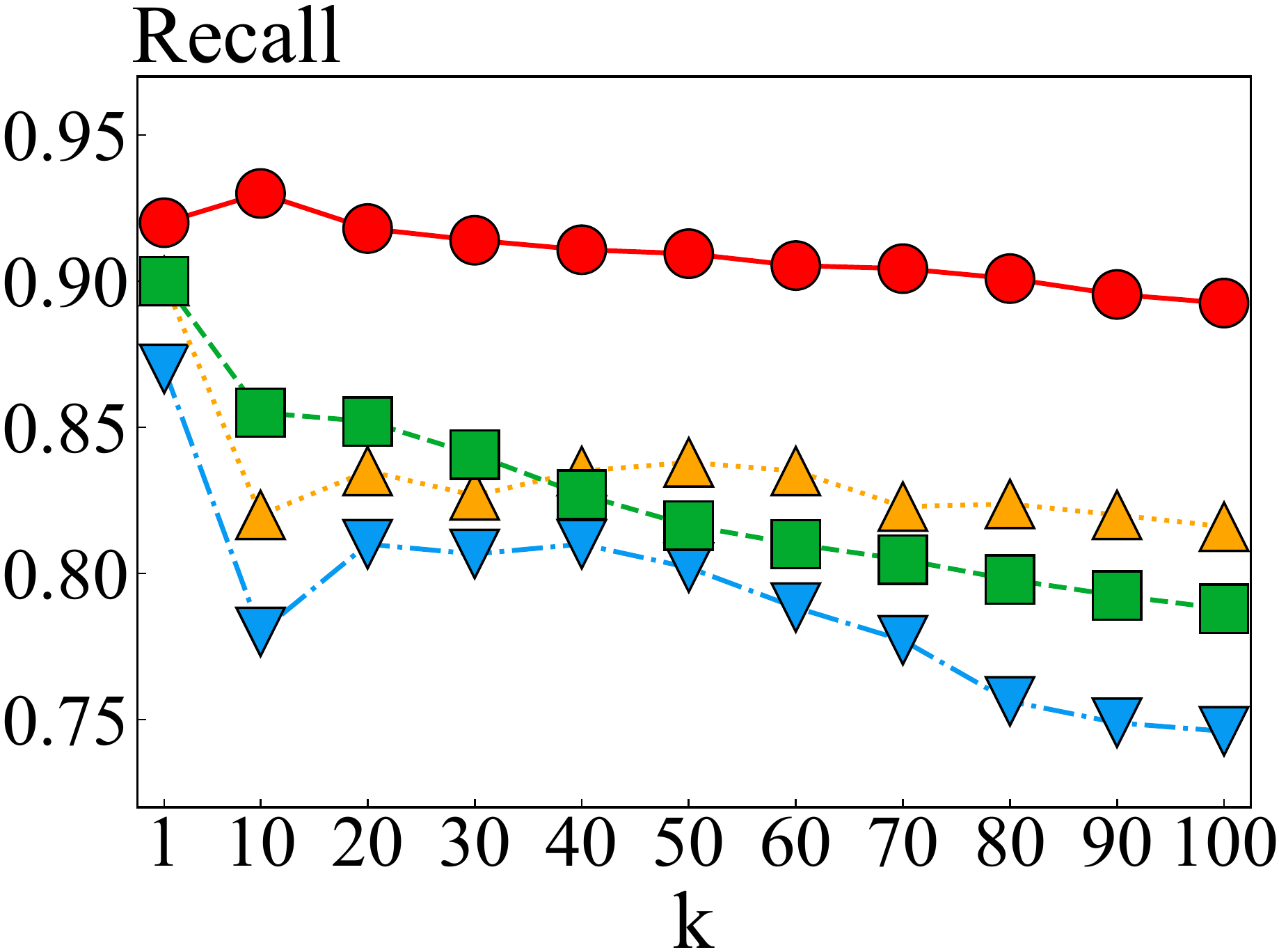}
				\vspace{-1.05em}
		\end{minipage}
	}
	\subfigure[OverRatio on TinyImages80M]{
		\begin{minipage}[c]{0.23\linewidth}
			\centering

			\includegraphics[width=1\textwidth]{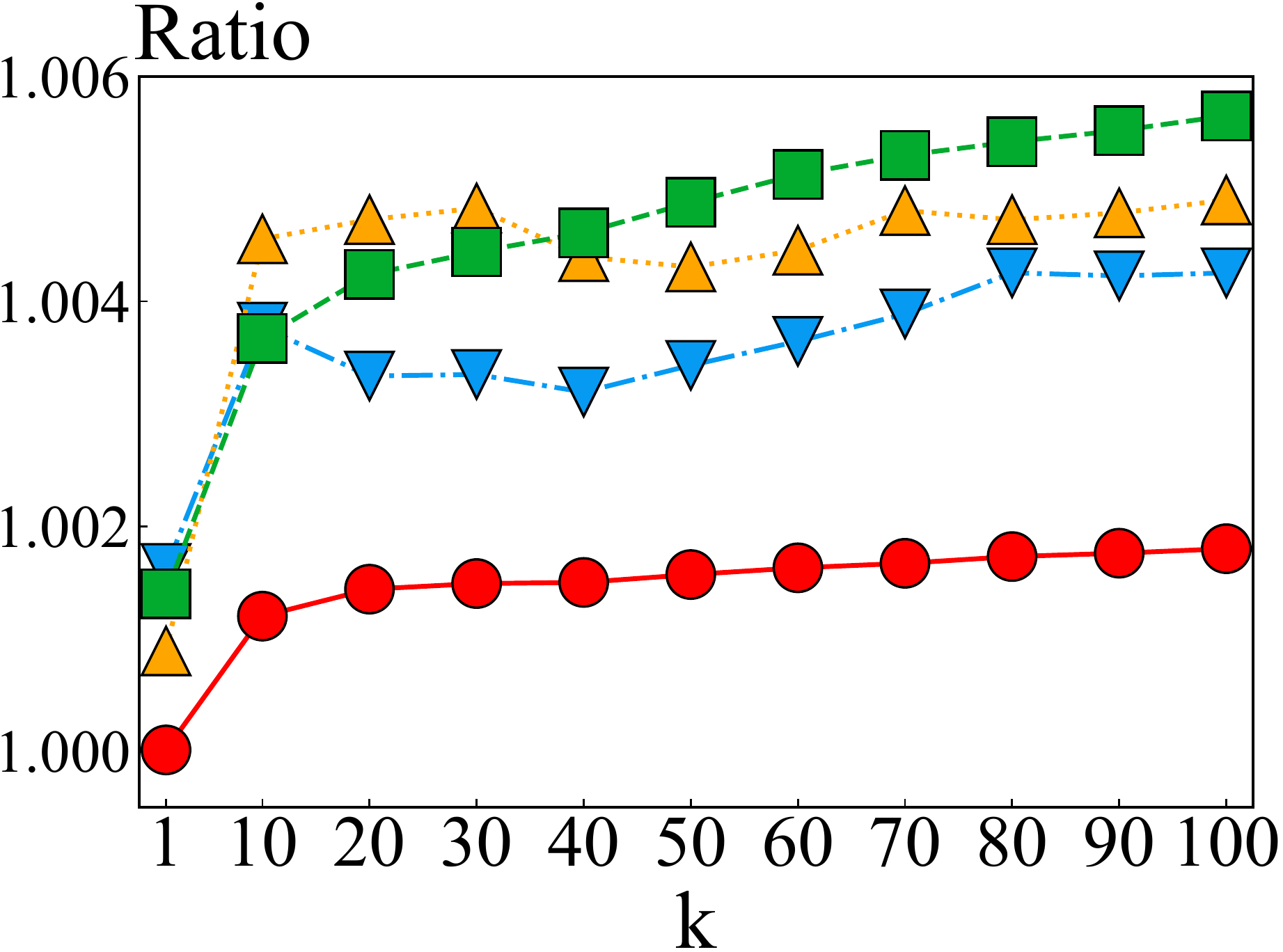}
				\vspace{-1.05em}
		\end{minipage}
	}
	\vspace{-0.4em}\caption{Performance when Varying $k$} \vspace{-1em}
	\label{expe:k}
\end{figure*}

\begin{figure*}[!h]
	\centering
	\subfigure[Trevi]{
		\begin{minipage}[c]{0.23\linewidth}
			\centering
			\includegraphics[width=1\textwidth]{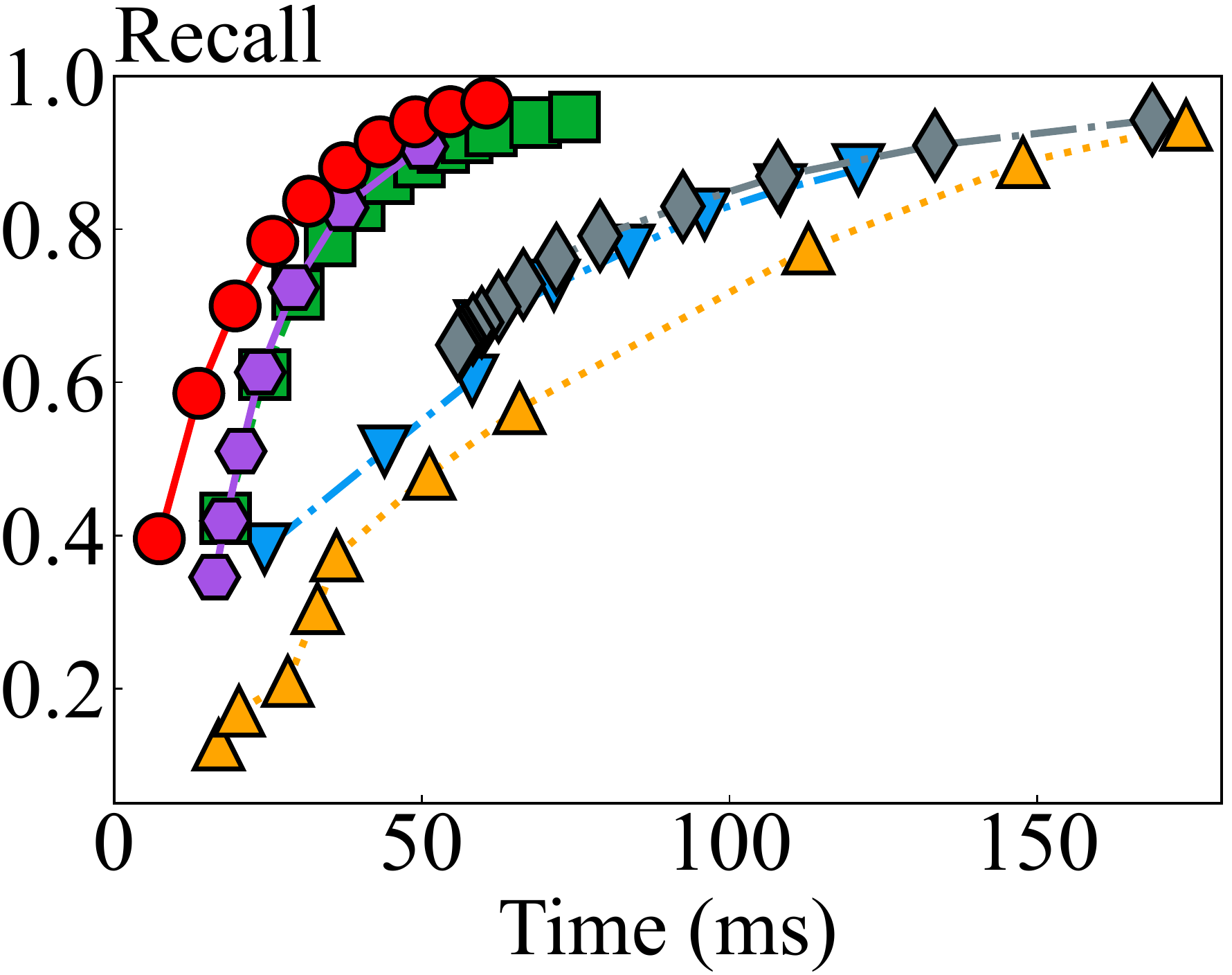}\vspace{0.1em}
		\end{minipage}
	}
	\subfigure[Gist]{
		\begin{minipage}[c]{0.23\linewidth}
			\centering
			\includegraphics[width=1\textwidth]{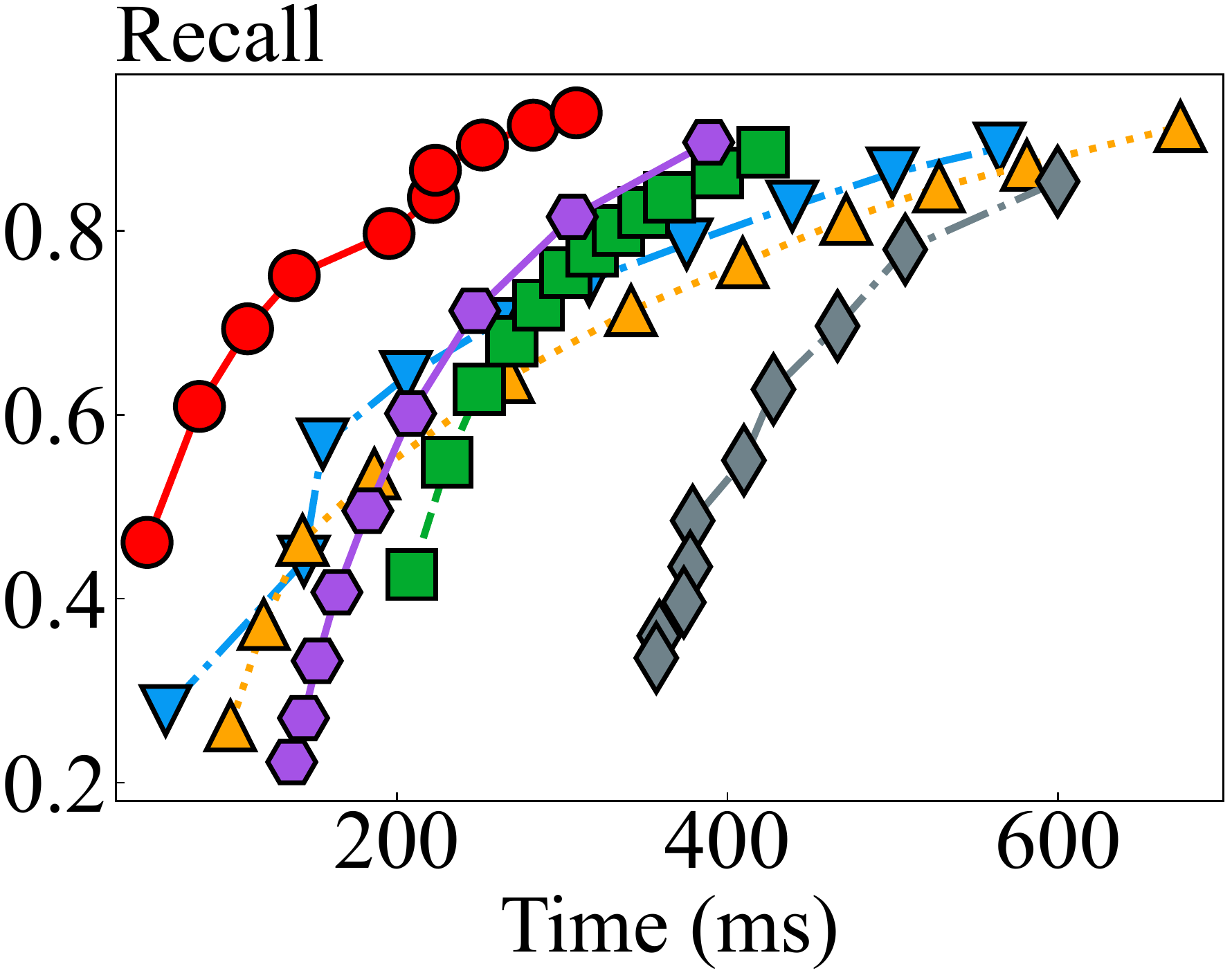}\vspace{0.1em}
		\end{minipage}
	}
	\subfigure[SIFT10M]{
		\begin{minipage}[c]{0.23\linewidth}
			\centering
			\includegraphics[width=1\textwidth]{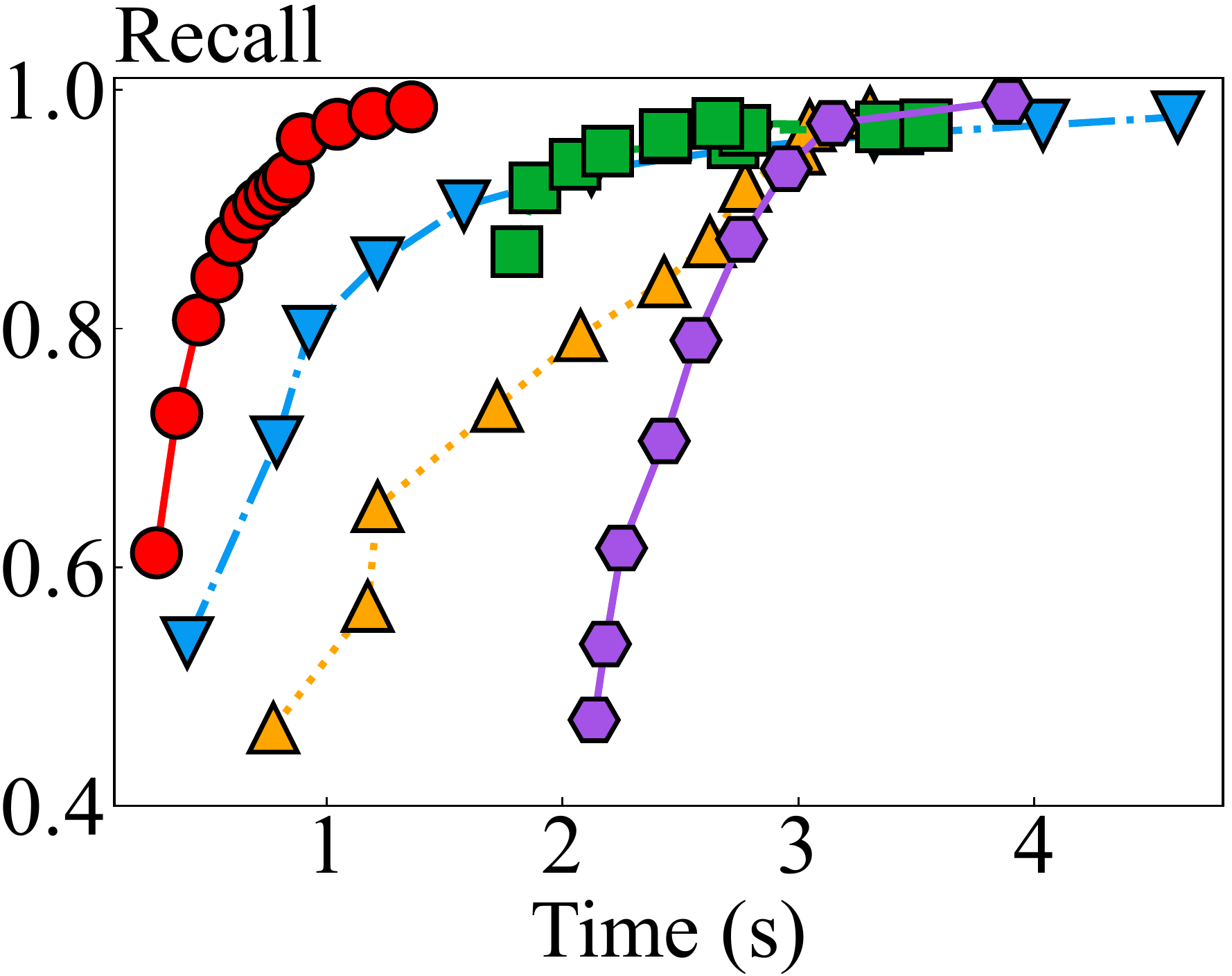}\vspace{0.1em}
		\end{minipage}
	}
	\subfigure[TinyImages80M]{
		\begin{minipage}[c]{0.23\linewidth}
			\centering
			\includegraphics[width=1\textwidth]{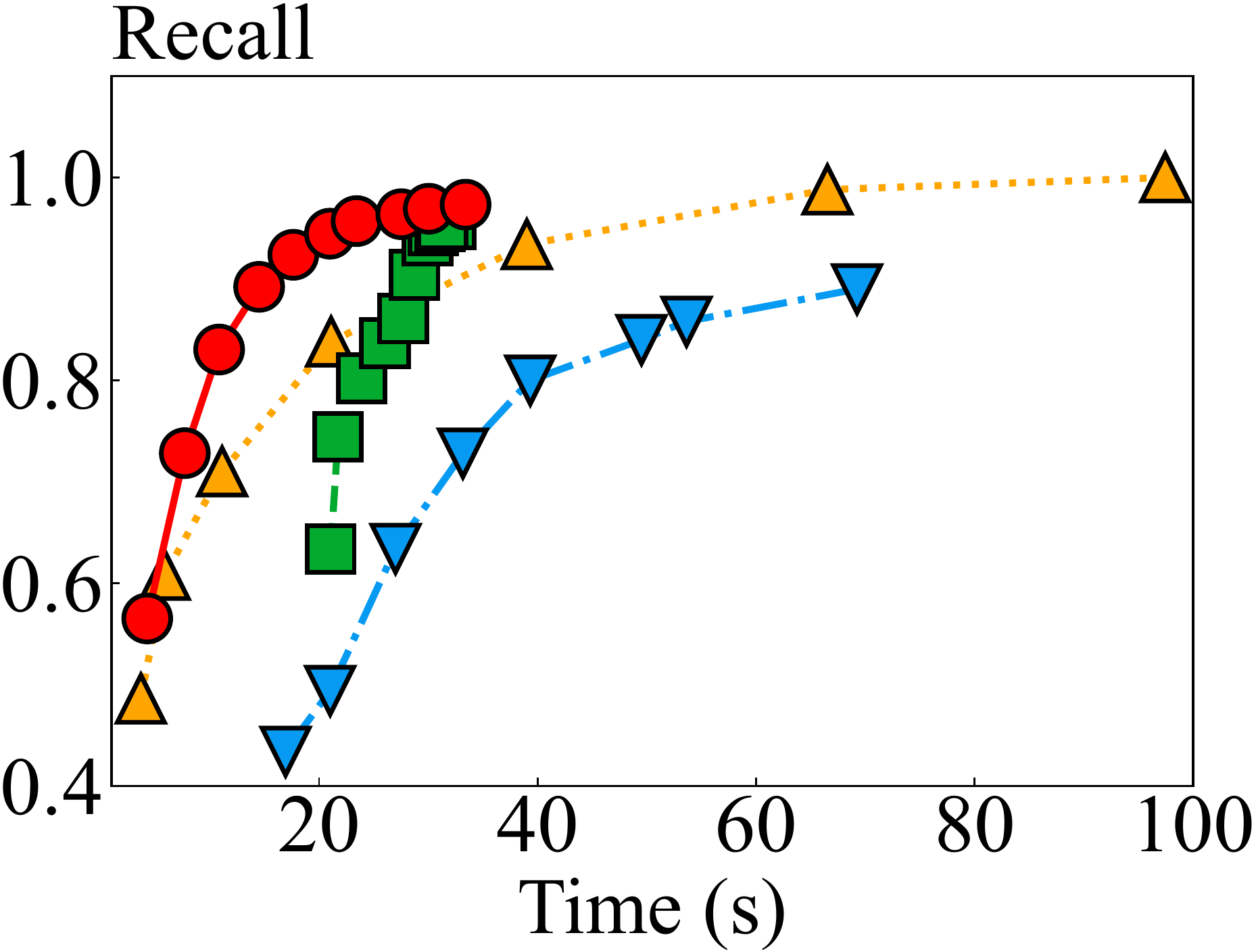}\vspace{0.1em}
		\end{minipage}
	}
	\vspace{-0.4em}
	\caption{Recall-Time Curves}\vspace{-1em}
	\label{expe:recall_time}
\end{figure*}
\begin{figure*}[!h]
	\centering
	\subfigure[Trevi]{
		\begin{minipage}[c]{0.23\linewidth}
			\centering
			\includegraphics[width=1\textwidth]{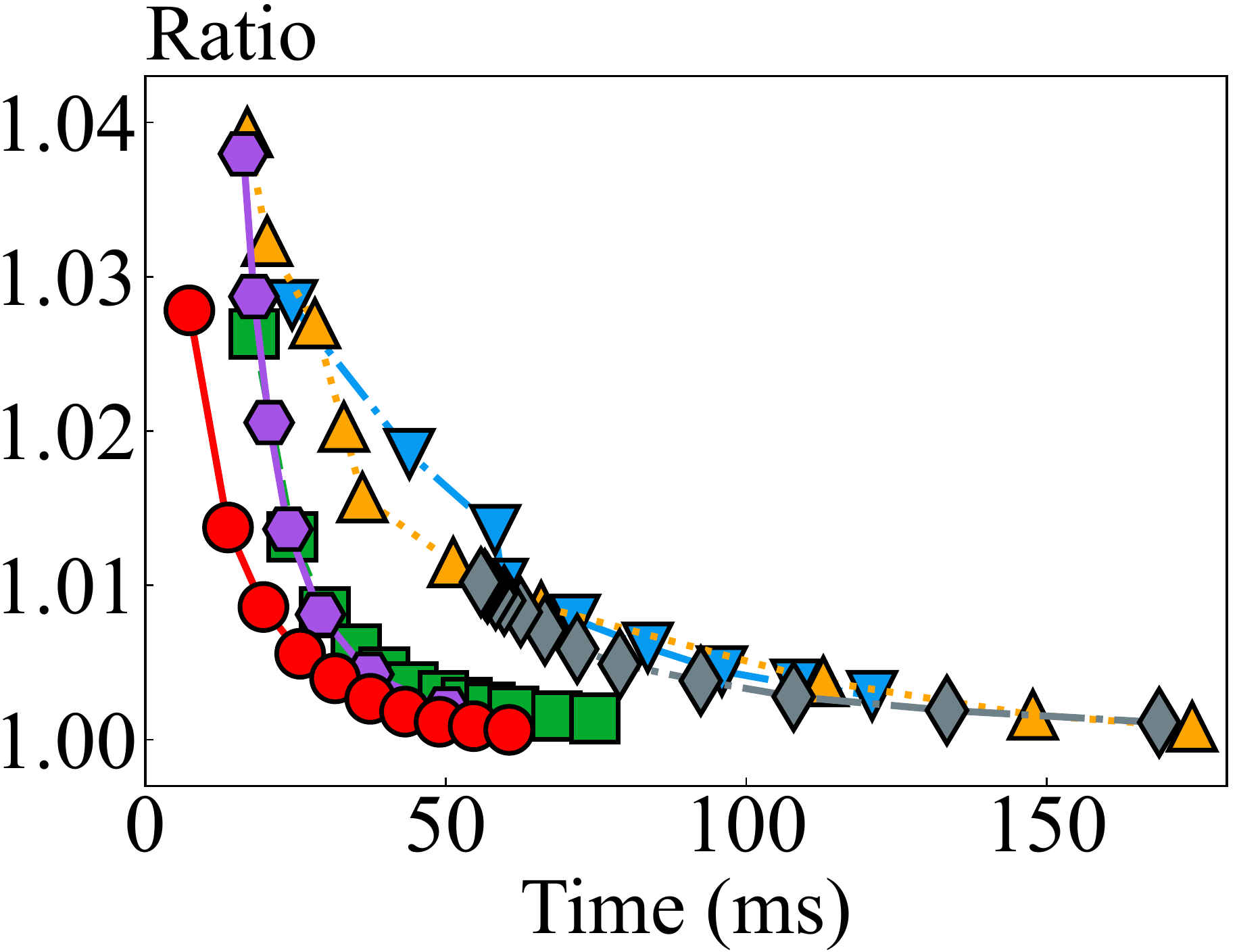}\vspace{0.1em}
		\end{minipage}
	}
	\subfigure[Gist]{
		\begin{minipage}[c]{0.23\linewidth}
			\centering
			\includegraphics[width=1\textwidth]{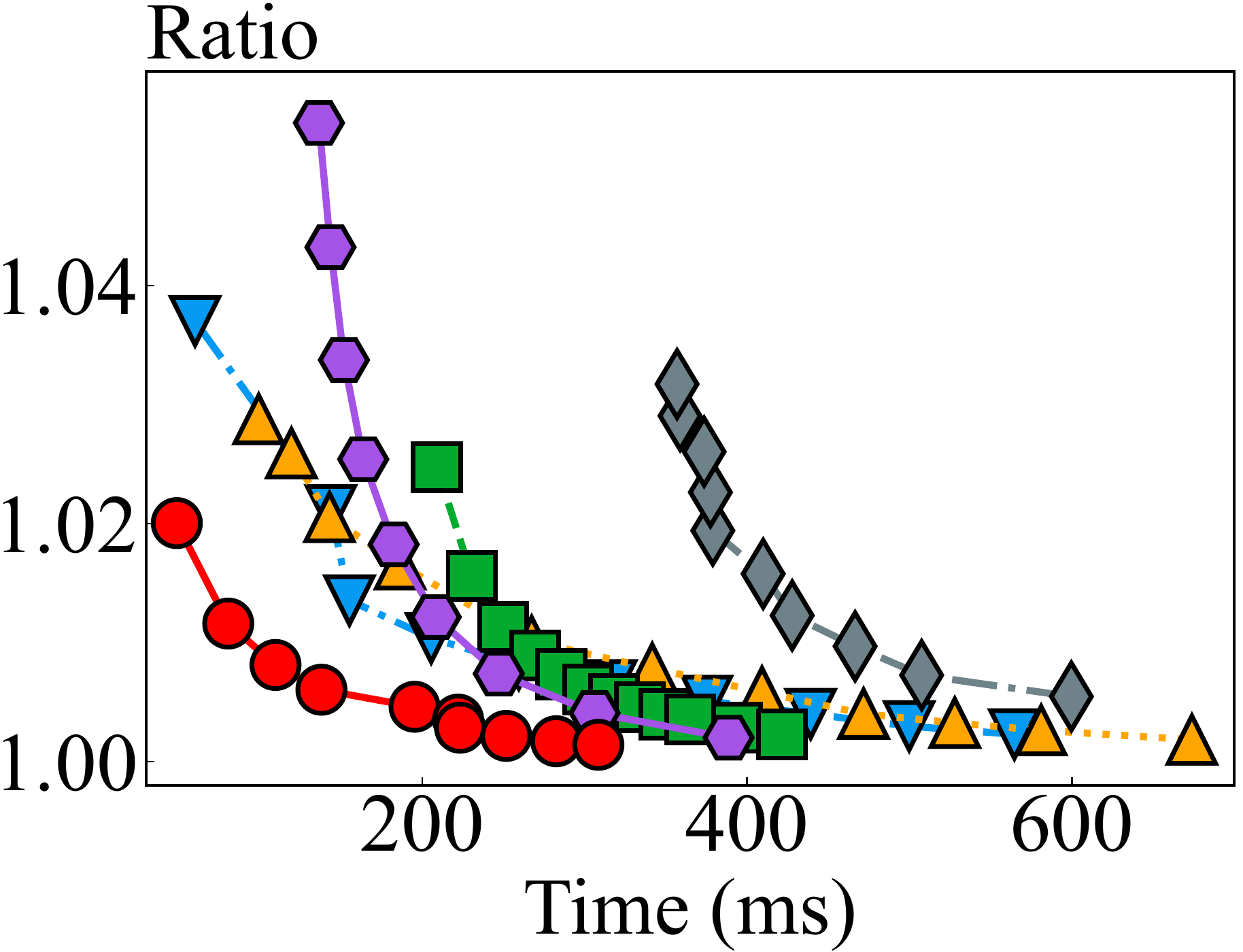}\vspace{0.1em}
		\end{minipage}
	}
	\subfigure[SIFT10M]{
		\begin{minipage}[c]{0.23\linewidth}
			\centering
			\includegraphics[width=1\textwidth]{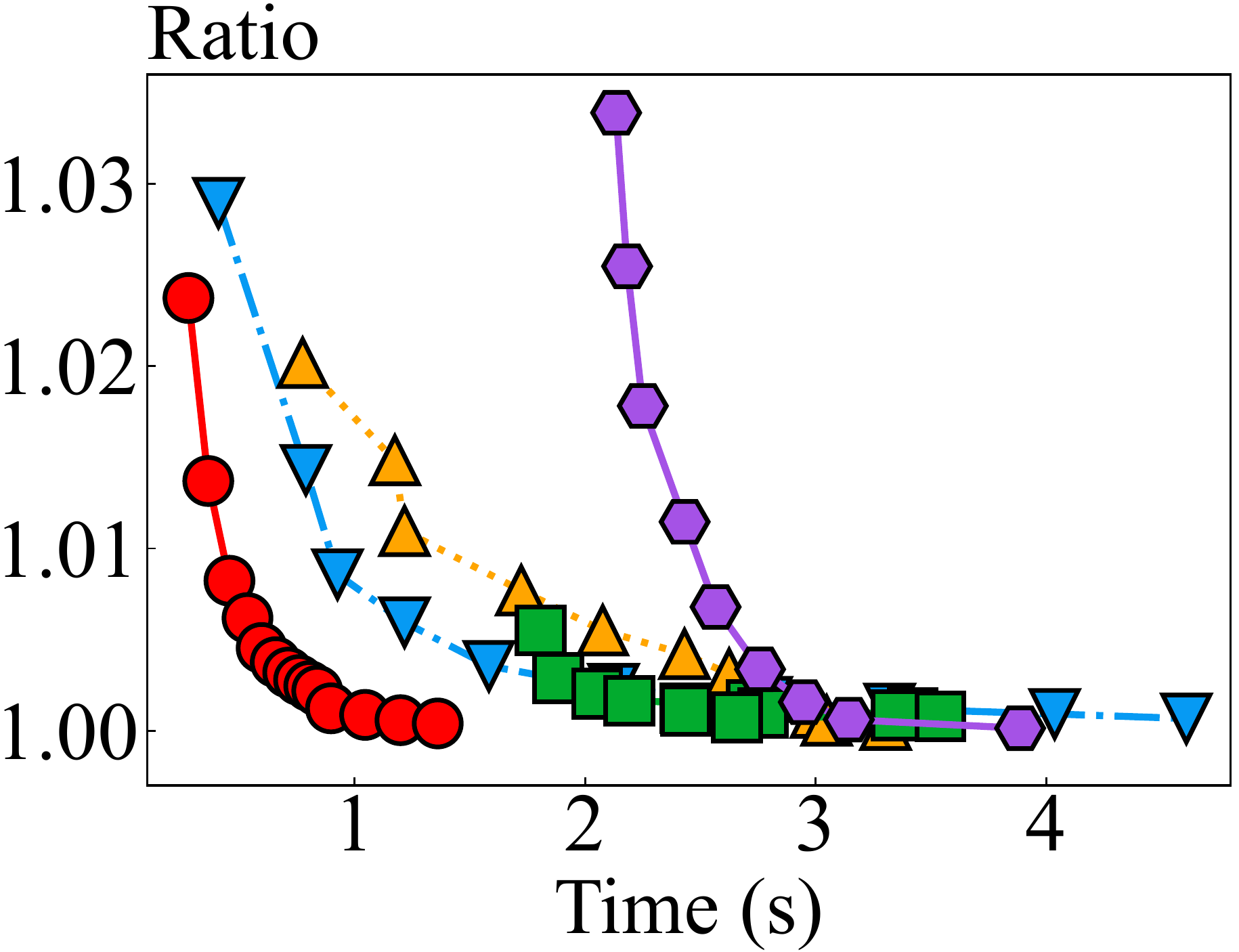}\vspace{0.1em}
		\end{minipage}
	}
	\subfigure[TinyImages80M]{
		\begin{minipage}[c]{0.23\linewidth}
			\centering
			\includegraphics[width=1\textwidth]{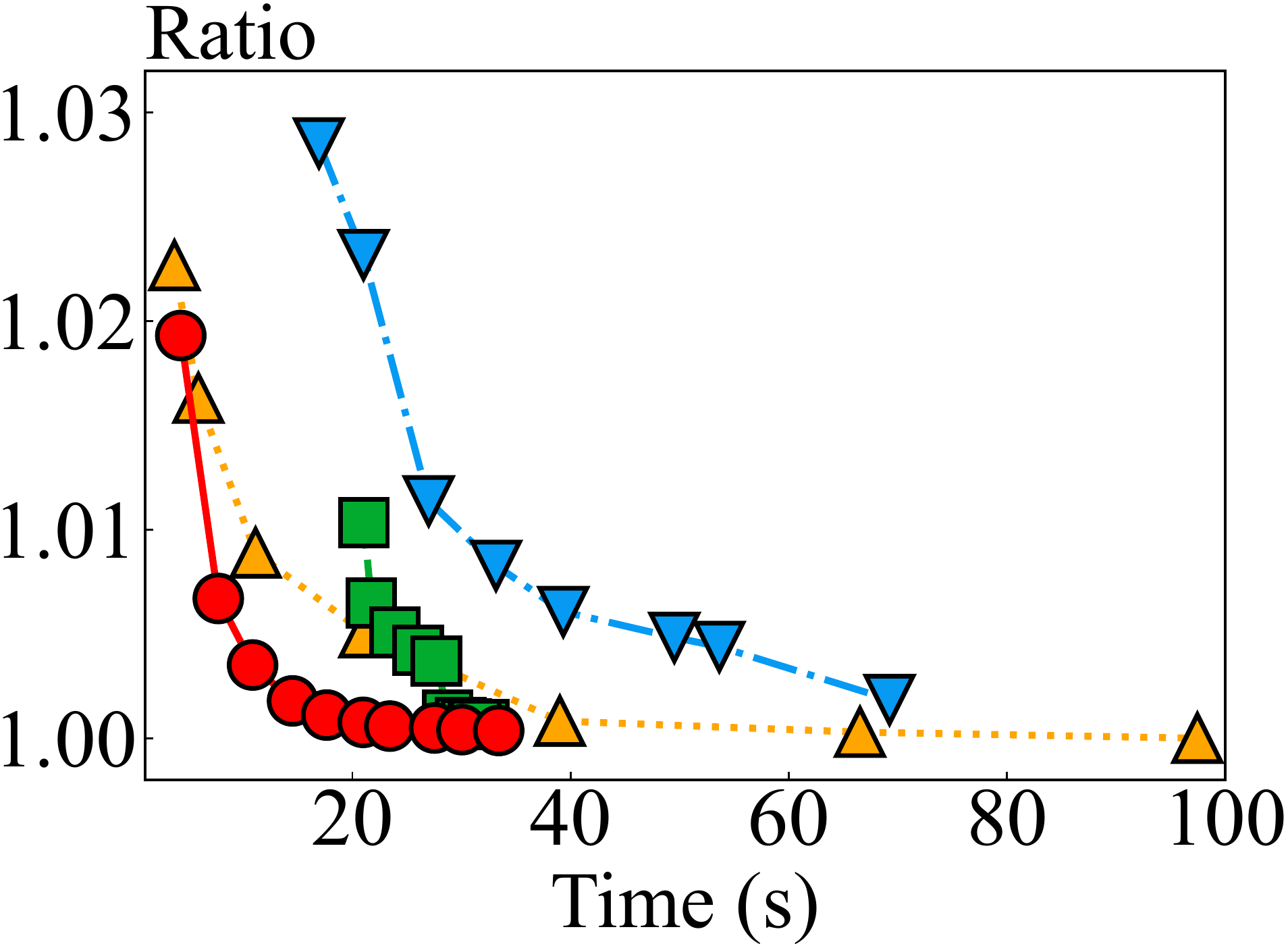}\vspace{0.1em}
		\end{minipage}
	}
	\vspace{-0.4em}
	\caption{Ratio-Time Curves}\vspace{-1em}
	\label{expe:ratio_time}
\end{figure*}


\subsubsection{\textbf{Effect of $k$}}

In this set of experiments, we study the query performance \cbl{in the default parameters when} varying $k$ in  $\{1, 10, 20, \cdots , 100\}$. 
Due to the space limitation, we only report recall and overall ratio on Gist and TinyImages80M in Figure \ref{expe:k}. \cthe{The} query time is omitted because the curve does not change much with $k$.
As expected, \cbl{DB-LSH} again yields the best accuracy,  \ie the \cbl{highest} recall and the \cbl{smallest} overall ratio. As $k$ increases, all algorithms have slightly worse accuracy because the average number of candidates checked for one result decreases, making the probability of missing some exact NNs slightly higher and thus affecting the accuracy.
At each $k$,
DB-LSH keeps outperforming the second best algorithms by an average of $5$-$10\%$ recall. 
Considering \cthe{the} smaller query time in DB-LSH (see Table \ref{tab:overview}), DB-LSH achieves better accuracy with higher efficiency for all $k$.

\subsubsection{\textbf{Recall-Time and OverallRatio-Time Curves}}

In this set of experiments, we plot the recall-time and overall ratio-time curves \cbl{by varying the approximation ratio $c$} for all algorithms to get a
\cbl{complete picture of the trade-off between the query efficiency and query accuracy}. Figures \ref{expe:recall_time}-\ref{expe:ratio_time} present the results on datasets Trevi, Gist, SIFT10M and TinyImages80M.
From the figures, we have \cbl{the} following observations:
(1) Among all algorithms, DB-LSH takes the least time to reach the same recall or overall ratio, which indicates DB-LSH achieves the best trade-off between accuracy and efficiency. The reason is that DB-LSH requires the fewest number of candidates to be accessed for a given accuracy among all algorithms. 
Compared to the second best algorithm on different datasets, DB-LSH can reduce the query time by $10$-$70$\% for given recall.
(2) DB-LSH always performs best, but the second best algorithms vary: R2LSH and PM-LSH on Trevi, R2LSH and FB-LSH on Gist, PM-LSH and FB-LSH on SIFT10M, PM-LSH and LCCS-LSH on TinyImages80M. R2LSH performs well on small datasets but becomes worse on larger datasets.
(3) As the query time increases, all algorithms return more accurate results, which is in line with the philosophy of LSH methods, \ie trading accuracy for efficiency.

\section{Conclusion} \label{sec:conclusion}
In this paper, we have proposed a novel LSH approach called DB-LSH for approximate nearest neighbor query processing in high-dimensional spaces with strong theoretical guarantees.
By decoupling the hashing and bucketing processes of the $(K,L)$-index and managing the projected points with a multi-dimensional index, DB-LSH can significantly reduce index size.
A query-centric dynamic bucketing strategy has been developed to avoid the hash boundary issue and thus generate high-quality candidates. We proved theoretically that DB-LSH can achieve a smaller $\rho^*$ bounded by $1/c^{4.746}$ when the initial bucket size is $w_0=4c^2$, which enables us to simultaneously reduce the query processing time and index space complexity of $(K, L)$-index methods.
A thorough range of experiments showed that DB-LSH comprehensively outperforms all competitor algorithms in terms of both efficiency and accuracy without the need for large indexes. DB-LSH can reduce the query time by an average of $40$\% compared to the second best competitors.
In the future, we plan to improve this work further by considering more efficient search strategies and early termination conditions.


\section*{Acknowledgment}
This research is partially supported by The Hong Kong Jockey Club Charities Trust.



\bibliographystyle{abbrv}
\bibliography{myRef}
\end{document}